\documentclass[journal,10pt,twocolumn]{IEEEtran}
\usepackage[final]{graphicx}
\graphicspath{{}} 

\usepackage{psfrag}
\usepackage{epstopdf}
\usepackage{amsfonts}
\usepackage{mathrsfs,amssymb}
\usepackage{pifont}
\usepackage{verbatim}
\usepackage{upgreek}
\usepackage{color}
\usepackage[usenames,dvipsnames]{xcolor}
\usepackage[font=footnotesize]{caption}
\usepackage{algorithmic}
\usepackage{algorithm}
\usepackage{mdwmath}
\usepackage{mdwtab}
\usepackage{amssymb,amsmath}
\usepackage{amsmath}
\usepackage{amsthm}
\usepackage{epstopdf}
\usepackage{cite}
\usepackage{tikz}
\usepackage{scalefnt}
\usepackage{subfigure}
\usepackage{MnSymbol}
\usepackage{setspace}
\usepackage{mathtools}
\usepackage{psfrag}

\label{newcommand}

\label{English Chars}
\newcommand{\Ss}{\mathcal{S}}
\newcommand{\Sc}{\mathcal{S}^{\rm c}}

\newcommand{\Bs}{\mathcal{B}}
\newcommand{\Us}{\mathcal{U}}

\newcommand{\xs}{\{x\}}
\newcommand{\ys}{\{y\}}
\newcommand{\FP}{{\rm FP}}
\newcommand{\bA}{\textbf{A}}
\newcommand{\ba}{\textbf{a}}

\newcommand{\bb}{\textbf{b}}
\newcommand{\bC}{\textbf{C}}

\newcommand{\be}{\textbf{e}}
\newcommand{\bF}{\textbf{F}}

\newcommand{\bI}{\textbf{I}}

\newcommand{\bJ}{\textbf{J}}

\newcommand{\bs}{\textbf{s}}

\newcommand{\bU}{\textbf{U}}

\newcommand{\bV}{\textbf{V}}

\newcommand{\bW}{\textbf{W}}
\newcommand{\bw}{\textbf{w}}

\newcommand{\by}{\textbf{y}}

\newcommand{\bz}{\textbf{z}}

\newcommand{\tr}{\mbox{\boldmath tr}\, }

\newcommand{\diag}{\mbox{\boldmath diag}\, }
\renewcommand{\det}{\mbox{\boldmath det}\, }
\renewcommand{\log}{\mbox{\boldmath log}\, }

\newcommand{\vect}{\mbox{\boldmath vec}\, }

\label{Roman Chars}

\newcommand{\btheta}{\mbox{\boldmath{$\theta$}}}

\newcommand{\bfi}{\mbox{\boldmath{$\phi$}}}

\label{Theorems}
\newtheorem{theorem}{Theorem}[section]

\theoremstyle{remark}

\begin{document}
\label{title}
\title{Sparse Antenna and Pulse Placement for Colocated MIMO Radar}
\author{Ehsan~Tohidi\IEEEauthorrefmark{1}, Mario~Coutino\IEEEauthorrefmark{2}, Sundeep~Prabhakar~Chepuri\IEEEauthorrefmark{2}, Hamid~Behroozi\IEEEauthorrefmark{1},\\ Mohammad~Mahdi~Nayebi\IEEEauthorrefmark{1}, and Geert~Leus\IEEEauthorrefmark{2}
	
\thanks{\IEEEauthorrefmark{1}Department of Electrical Engineering, Sharif University of Technology, Tehran, Iran.
	E-mails: Tohidi@ee.sharif.edu,
	Behroozi@sharif.edu,
	Nayebi@sharif.edu}
\thanks{\IEEEauthorrefmark{2}Faculty of Electrical Engineering, Mathematics and Computer Science, Delft University of Technology, Delft, The Netherlands.
E-mails: M.A.CoutinoMinguez-1@tudelft.nl,
S.P.Chepuri@tudelft.nl,
G.J.T.Leus@tudelft.nl 
}
\thanks{This manuscript is an extension of the work presented in \cite{TohidiAsilomar2017}. This work is supported in part by the KAUST-MIT-TUD consortium under grant OSR-2015-Sensors-2700 and the ASPIRE project (project 14926 within the STW OTP programme), which is financed by the Netherlands Organisation for Scientific Research (NWO). Mario Coutino is partially supported by CONACYT.}
}

\markboth{Tohidi et al. Sparse Antenna and Pulse Placement for Colocated MIMO Radar}%
{Shell \MakeLowercase{\textit{et al.}}: Bare Demo of IEEEtran.cls for IEEE Communications Society Journals}

\maketitle

\begin{abstract}
\boldmath 
Multiple input multiple output (MIMO) radar is known for its superiority over conventional radar due to its antenna and waveform diversity. Although higher angular resolution, improved parameter identifiability, and better target detection are achieved, the hardware costs (due to multiple transmitters and multiple receivers) and high energy consumption (multiple pulses) limit the usage of MIMO radars in large scale networks.  
On one hand, higher angle and velocity estimation accuracy is required, but on the other hand, a lower number of antennas/pulses is desirable. To achieve such a compromise, in this work, the Cram\'er-Rao lower bound (CRLB) for the angle and velocity estimator is employed as a performance metric to design the antenna and pulse placement. It is shown that the CRLB derived for two targets is a more appropriate criterion in comparison with the single-target CRLB since the two-target CRLB takes into account both the mainlobe width and sidelobe level of the ambiguity function. In this paper, several algorithms for antenna and pulse selection based on convex and submodular optimization are proposed. Numerical experiments are provided to illustrate the developed theory.
\end{abstract}

\section {Introduction}
Multiple input multiple output (MIMO) radar has been gaining a lot of interest during the last decade \cite{5393291}. The main reason behind this growth is the enormous capabilities that this type of radar provides, e.g., higher angular resolution, improved parameter identifiability, and radar cross section (RCS) diversity \cite{4350230,4408448}. Based on the antenna configuration MIMO radars are categorized, into {\it colocated} and {\it widely-separated} MIMO radars. Colocated MIMO radars have closely located antennas, which see the targets from the same angle. A high angular resolution due to waveform diversity is one of the main advantages of colocated MIMO radars \cite{7126203,6957532,6650099,6132420,5752847,5672411,6678708}. The other category, widely-separated MIMO radars, have transmitter/receiver antennas placed far from each other. This results in different target angles of view for different transmitter-receiver pairs. Low speed moving target detection due to the spatial diversity gain is among the advantages widely-separated MIMO radars (see \cite{5393291} and \cite{7771508,6621849,5989873}). In this paper, the focus is on the colocated MIMO radar configuration to estimate the angle and velocity of the targets (the developed design algorithms can be easily adapted to widely-separated MIMO radars as discussed later).

Angle of arrival and velocity estimation are the main tasks of any radar system~\cite{ender2010compressive}. Due to the additional degrees of freedom, MIMO radars perform these tasks much better than a single radar \cite{godrich2010target,5393291,5672411}. In \cite{he2016generalized} and \cite{ai2015cramer} the Cram\'er-Rao lower bound (CRLB) for a MIMO radar has been derived to prove this advantage. Beside the numerous advantages of MIMO radars over conventional radars, the main drawbacks of these radars are, however, the large hardware costs due to multiple transmitter and receiver chains, the high energy consumption due to multiple transmitted pulses, and the large computational complexity involved in processing the transmitted pulses. To reduce these costs, keeping in mind the low number of targets in the region of interest, compressive sensing (CS) based approaches have shown promising performance \cite{ender2010compressive,7330290}. Although CS-based approaches reduce the number of measurements to be processed, the hardware costs are not reduced. This is because of the dense sampling matrices used in CS that limit the number of measurements while requiring all the antennas and pulses. Alternatively, antenna and pulse selection (i.e., employing only a subset of all the antennas and pulses) via sparse sensing can be performed to reduce the hardware sensing costs as well as the energy consumption, while achieving the desired performance. We would like to stress here that pulse placement for radar has been rarely considered before. A closely related topic is waveform design \cite{4350230} which deals with the design and selection of transmit waveforms with proper characteristics. However, such designs are mainly concerned with statistical properties of the signal within each pulse rather than in the selection of the positions of the transmit pulses within the pulse sequence.

Sensor selection is the problem of choosing a subset of sensors out of a set of candidate sensors. Sensor selection is important to reduce the hardware costs, computational complexity, network energy consumption, and has been studied vastly as detailed next. A knapsack problem formulation for sensor selection is proposed in \cite{godrich2012sensor}, where an algorithm based on a greedy heuristic is presented. Sensor selection via convex optimization is proposed in \cite{joshi2009sensor}, where the problem is first relaxed to a convex program, and then, sensors are selected through solving a convex optimization problem. Similarly, \cite{roy2013sparsity} proposed a sparsity-enforcing sensor selection scheme for direction of arrival estimation, where a single-target CRLB is used as objective function with additional constraints on the sidelobe level. In \cite{chepuri2015sparsity} and \cite{rao2015greedy}, sensor selection for general non-linear models through convex and submodular optimization, respectively, is proposed.
 
Antenna and pulse selection can be posed as a sensor selection problem where a subset of antennas and pulses is selected out of a large number of antennas and pulses. We refer to this problem as antenna and pulse placement. Antenna placement in widely separated MIMO radar for joint target position and velocity estimation is studied in \cite{5393291,greco2011cramer,ivashko2017radar}, which are all based on the single-target CRLB. For instance, \cite{rossi2014spatial} proposed a DOA estimation framework for a MIMO radar in which transmit and receive antenna positions are drawn at random from a uniform distribution. In a similar way, by employing the single-target CRLB in \cite{4487196}, joint array and waveform optimization techniques for MIMO radar are investigated. The authors in \cite{4487196} show that both local and global errors incurred by the estimator must be considered during the design phase. In fact, these effects occur in low SNR scenarios when the estimator exhibits a threshold effect due to local maxima in the ambiguity function. However, interactions between two or more targets is not considered in \cite{4487196}. This is achieved in this work by considering the general expression of the two-target CRLB.

Typically, radars transmit several pulses with a uniform time separation, which is called the pulse repetition interval (PRI). By exploiting the phase differences of the reflected pulses from the targets, Doppler (or velocity) estimation is performed \cite{skolnik2008radar}. A velocity estimation algorithm for wide-band frequency-modulated continuous-wave radar systems using the phase differences of consecutive uniformly separated pulses is proposed in \cite{wagner2013wide}. To reduce the network energy consumption and processing costs, we aim to have an irregular pulse transmission pattern (i.e., by transmitting only a subset of the uniformly separated pulses). Figure \ref{pulse_selection_sample} shows an example of such an irregular pulse placement.
Similar to pulse placement, the idea behind antenna placement is to perform the angle of arrival estimation task with a smaller number of antennas. In colocated MIMO radars, transmitters and receivers are usually placed uniformly along a line with a spacing of half a wavelength. However, we want to systematically design the transmitter-receiver positions to obtain a nonuniform array with a reduced number of transmit/receive elements. In particular, we start with a large set of candidate locations where we can place the antennas. Then, we select the best subset out of those locations in order to achieve a desired estimation performance. This antenna placement procedure helps to reduce the hardware costs and computational complexity, while maintaining a prescribed performance. Figure \ref{antenna_selection_sample} illustrates an irregular transmitter and receiver placement in comparison with a uniform placement. 
\begin{figure}
\centering
\psfrag{irregular pulse placement}{\scriptsize{irregular pulse placement}}
\psfrag{uniform pulse placement}{\scriptsize{uniform pulse placement}}
\psfrag{irregular antenna selection}{\scriptsize{irregular antenna placement}}
\psfrag{uniform antenna selection}{\scriptsize{uniform antenna placement}}
\psfrag{Transmitters}{\scriptsize{Transmitters}}
\psfrag{Receivers}{\scriptsize{Receivers}}
\subfigure[] {\includegraphics[width=.4\textwidth]{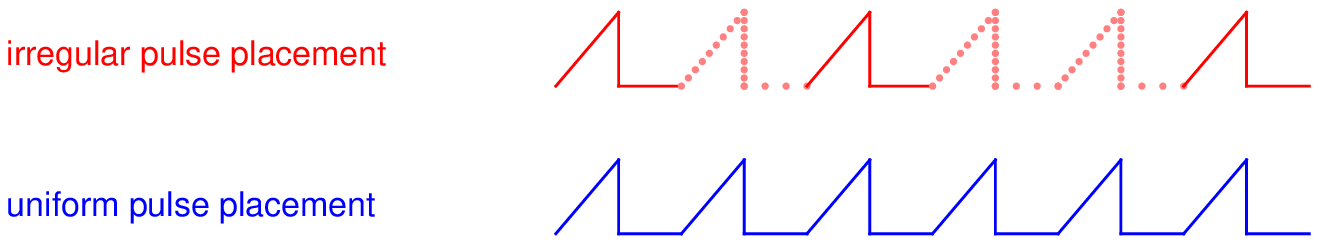}
\label{pulse_selection_sample}} \quad
\subfigure[] {\includegraphics[width=.4\textwidth]{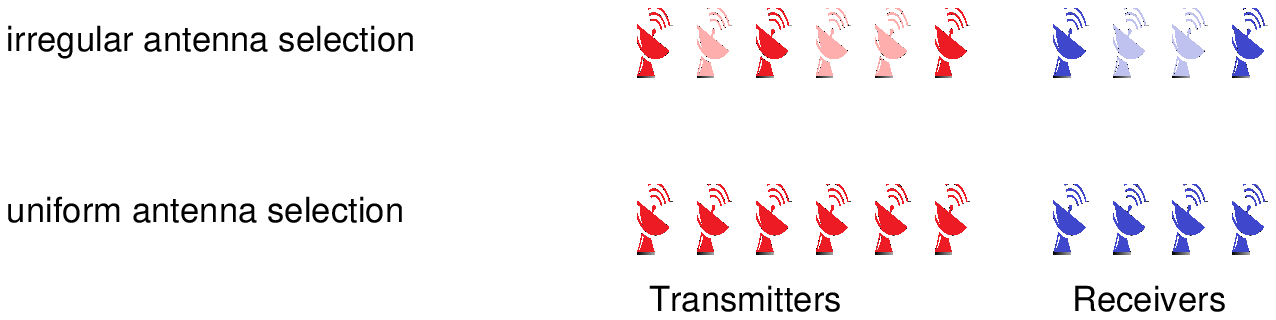}
\label{antenna_selection_sample}} \\
\caption{ \footnotesize{(a) Uniform vs. irregular pulse transmission. Dashed lines indicate that the pulses are not transmitted in that interval, (b) Uniform vs. irregular antenna placement. Light colored antennas indicate that the antennas are not used during transmission or reception.}}
\label{}
\end{figure}

The aim of this paper is to find the optimal antenna and pulse placement that guarantees a desired angle of arrival and velocity estimation accuracy. It should be noted that the performance is traded off with cost when the number of antennas and/or pulses are reduced. 
\subsection{Contributions}
In this paper, we study further the joint antenna and pulse placement for a colocated MIMO radar for angle of arrival and velocity estimation based on sparse sensing \cite{chepuri2016sparse} by extending our previous work \cite{TohidiAsilomar2017} by a more detailed signal model, a derivation of the two-target CRLB, and a submodular optimization framework as a fast and reliable alternative approach to convex optimization.

The conventionally used performance measure, namely, the single-target CRLB only considers the mainlobe width of the ambiguity function but does not take into account the sidelobe level. Therefore, we derive the CRLB for two targets, which takes into account both the mainlobe width and sidelobe level (particularly the sidelobe level around the mainlobe) of the ambiguity function. Based on this two-target CRLB, we propose several performance measures and develop a number of algorithms for designing the optimal antenna and pulse placement of colocated MIMO radar systems. Firstly, single antenna pulse placement and single pulse MIMO radar antenna placement are presented as two specific cases of the problem. Then, we present the general case of joint antenna and pulse selection. Since the antenna-pulse selection is a combinatorial optimization, and is NP-hard \cite{joshi2009sensor}, we propose several suboptimal algorithms for solving the selection problem. One of the proposed approaches is based on submodular optimization. We prove the submodularity of the employed performance measure, which enables us to use a greedy algorithm to perform the selection with near-optimality guarantees. The second proposed approach is based on convex optimization, where by employing some relaxations, the optimization is turned into a convex program. However, due to these relaxations, a suboptimal solution is obtained in general. The advantages and disadvantages of these two approaches are explained in more detail in Section IV.

\subsection{Outline and Notations}
The rest of the paper is organized as follows. In Section~II, the signal model is introduced. Section III provides the required preliminaries for this paper. The problem formulation is discussed in Section IV. Two basic examples to illustrate the concept are presented in Section V. The proposed algorithms for the most general form of the antenna-pulse selection problem is presented in Section VI. Simulation results are reported in Section VII. Finally, Section VIII concludes the paper.

We adopt the notation of using boldface lower (upper) case for vectors $\ba$ (matrices $\bA$). The transpose, Hermitian, and complex conjugate operators are denoted by the symbols $(.)^T$, $(.)^H$, and $(.)^*$, respectively. $\mathbb{R}^{N \times M}$ is the set of $N\times M$ real matrices. $\diag(\ba)$ indicates the diagonal matrix formed by the components of vector $\ba$ along the main diagonal. $\det\{.\}$ is the matrix determinant and $\tr\{.\}$ is the matrix trace operator. In addition, $||\bA||_0$ and $||\ba||_0$ are the number of non-zero entries of $\bA$ and $\ba$, respectively. If $\ba$ and $\bb$ are two vectors, then $\left<\ba,\bb\right>$ is the inner product between $\ba$ and $\bb$ (i.e., $\left<\ba,\bb\right> = {\ba^H\bb}$). $\lambda_{\rm max}\{\bA \}$ and $\lambda_{\rm min}\{ \bA \}$ are the maximum and minimum eigenvalues of the matrix $\bA$, respectively. Given a reference set $\mathcal{U}$ and a subset $\mathcal{A}\subseteq\mathcal{U}$, the absolute complement of $\mathcal{A}$ is denoted as $\mathcal{A}^\mathsf{c}$, i.e., $\mathcal{A}^\mathsf{c}$ = $\mathcal{U} \setminus \mathcal{A}$.

\section{Signal Model}
Consider a colocated MIMO radar with $R$ receivers placed along a line at coordinates $[d,2d,...,Rd]$, where $d$ is the inter-element spacing, which is assumed to be $\lambda/2$ ($\lambda$ is the wavelength). In addition, $I$ transmitters are placed along a line at coordinates $[-d,-2d,...,-Id]$. The $i$th transmitter can transmit a waveform $s_i(t)$ for $P$ times with PRI $T_P$. Note that $s_i(t)$ is non-zero only in the interval $[0,T_P]$. Assume $Q$ targets exist in the region of interest. Our aim is to estimate the position and velocity of the targets based on the received signals. In particular, the direction cosine, $u_q=cos\:\psi_q$ with $\psi_q$ being the $q$th target's angle of arrival and radial velocity, $v_q$ of the $q$th target, are the desired unknown parameters.

The noiseless baseband representation of the signal received at the $r$th receiver during the time interval $[pT_P,(p+1)T_P]$ due to all the targets is

\begin{equation}
x_{r}(t) = \sum_{q=1}^Q{\alpha_qh(t;v_q)\bs_{r,q}^T(t-pT_P)\bfi_{r}(u_q)},
\end{equation}
where $\alpha_q$ is the effect of the $q$th target's RCS and $h(t;v_q) = \exp(j4\pi{v_q}t/{\lambda})$ with $2v_q/\lambda$ being the Doppler frequency of the $q$th target. Due to the colocated configuration assumption, for each target, the Doppler frequency and RCS seen by all transmitter-receiver pairs are equal. Here, the assumption is that the RCS and propagation attenuation are constant during the observation interval (i.e., Swerling I model). Furthermore, $\bs_{r,q}(t) = [s_1(t-\tau_{1,q,r}),...,s_I(t-\tau_{I,q,r})]^T\in \mathbb{C}^I$ includes the received signal from all the transmitters, where 
\begin{equation} \nonumber
\begin{aligned}
\tau_{i,q,r} &= c^{-1}[(R_q-d_i\cos\psi_q)+(R_q-d_r\cos\psi_q)] \\&= c^{-1}[2R_q - (d_i+d_r)u_q]
\end{aligned}
\end{equation}
is the time delay of signal propagation between the $i$th transmitter, $q$th target, and $r$th receiver, and $c$ the speed of light. In addition, 
\begin{equation} \nonumber
\begin{aligned}
\bfi_{r}(u_q) &= [\exp(-j2\pi f_c \tau_{1,q,r}),...,\exp(-j2\pi f_c \tau_{I,q,r})]^T\in \mathbb{C}^I
\end{aligned}
\end{equation}
contains the related phase shifts with $f_c$ being the carrier frequency. In the expressions for $\bs_{r,q}(t)$ and $\bfi_{r}(u_q)$, $R_q$ is the $q$th target distance from the center of the coordinate system, and $d_i$ and $d_r$ are the positions of transmitter $i$ and receiver $r$ on the x-axis, respectively.

Linear frequency modulation (LFM) is selected for signaling and throughout the paper, orthogonal waveforms are transmitted by different transmitters. In fact, we adopt a set of LFM signals that have the same shape, but are slightly shifted in time which yields an efficient orthogonal transmission scheme \cite{6557996}. More specifically, we design the $i$th transmitted waveform as $s_i(t) = s(t-(i-1)t_{sh})$, where $t_{sh}$ is the time shift between adjacent LFM signals to achieve orthogonality and $s(t)$ is the baseband LFM waveform

\begin{equation}
s(t) = \exp\left(j\pi kt^2\right),
\end{equation}
where $k$ is the rate of sweeping the whole bandwidth for the pulse duration $T_C$, i.e., $k=B/T_C$ with $B$ being the signal bandwidth. Note that to satisfy the orthogonality condition, $t_{sh}$ should be selected larger than the time delay of the farthest target of interest. In addition, since all pulses from all antennas need to fit within a single PRI after reception, we need the condition $T_C+It_{sh}<T_P$. It should be pointed out that any set of orthogonal waveforms other than LFM signaling could also be employed for our model.

Employing $I$ matched filters (de-ramping plus filtering) matched to the $I$ transmit waveforms, the observed signal in the time interval $[pT_P,(p+1)T_p]$ from all the transmitters of the $r$th receiver $\bz_{r}(t) = [z_{r,1}(t),...,z_{r,I}(t)]^T$, after some simplifications is given by

\begin{equation}
\begin{aligned}
\bz_{r}(t) = \sum_{q=1}^Q{\alpha_qh(t;v_q) \beta_q(t-pT_P) \bfi_r(u_q)} + \be_{r}(t),
\end{aligned}
\end{equation}
where $\beta_q(t) = \rm{exp}$$\{j4\pi kt R_qc^{-1}\}$. Note that the entries of $\be_r(t)=[e_{r,1}(t),...,e_{r,I}(t)]^T$ are noise terms at the output of the $I$ matched filters at the $r$th receiver. Since the $I$ transmit waveforms are orthogonal, the entries of $\be_r(t)$ can be assumed independent. Thus, $e_{r,i}(t), i=1,...,I,$ are modeled as i.i.d. with distribution $\mathcal{N}(0,\sigma_e^2)$.

Sampling the observed signal with sampling period $T_s$, $N$ 
 samples per pulse are obtained, where the $n$th sample of the $p$th pulse related to the transmitter-receiver pair $(i,r)$, $z_{r,i,p}[n]$, is given by
\begin{equation}
\begin{aligned}
z_{r,i,p}[n] &= z_{r,i}(pT_p+nT_s) \\&=  \sum_{q=1}^Q{\alpha_qh(pT_P+nT_s;v_q) \beta_q(nT_s) \phi_{r,i}(u_q)} \\&+ e_{r,i}(pT_P+nT_s) \\
&=  \sum_{q=1}^Q{y_{r,i,p}^{(q)}[n]} + e_{r,i,p}[n] =  y_{r,i,p}[n] + e_{r,i,p}[n],
\end{aligned}
\label{equ:Qtarget}
\end{equation}
where $\phi_{r,i}(u_q)$ is the $i$th entry of $\bfi_r(u_q)$. Collecting all the measurements, we have a non-linear model of the form

\begin{equation}
\bz =  \by(\btheta) + \be \in \mathbb{C}^{NRIP},
\label{equ:vectorform}
\end{equation}
where the unknown parameters of the $q$th target are represented by the vector $\btheta_q = [u_q,v_q]^T\in\mathbb{R}^2$. So, $\btheta = [\btheta_1^T,...,\btheta_Q^T]^T\in \mathbb{R}^{2Q}$ collects all the unknown parameters.

\section{Performance metric}
As seen in Section II, the measurements are a non-linear function of the unknown parameters. As a result, the mean squared error (MSE) does not admit a closed form expression \cite{ranieri2014near}. On the other hand, the CRLB provides a lower bound on the variance of any unbiased estimator and can be used to evaluate the performance of unbiased estimators. Since the CRLB can always be computed in closed form, it is employed as an estimation performance criterion.

It is well known that under the regularity condition, the covariance of any unbiased estimator $\hat{\btheta}$ of the unknown vector ${\btheta}$ is lower bounded by the CRLB as \cite{Kay2008statistical,chepuri2015sparsity}:

\begin{equation}
\mathbb{E}\{(\btheta-\hat{\btheta})(\btheta-\hat{\btheta})^H\} \geq \bC(\btheta)=\bF^{-1}(\btheta),
\end{equation}
where $\bC$ is the CRLB matrix and $\bF$ is the Fisher information matrix (FIM), which can be calculated as \cite{chepuri2015sparsity}

\begin{equation}
\begin{aligned}
\bF(\btheta) &= -\mathbb{E}\left\{\frac{\partial^2 \ln \: p(\mathbf{z};\btheta)}{\partial\btheta\: \partial\btheta^H}\right\} \\&= \mathbb{E}\left\{\frac{\partial \ln \: p(\mathbf{z};\btheta)}{\partial\btheta}\frac{\partial \ln \: p(\mathbf{z};\btheta)}{\partial\btheta^H}\right\} \in \mathbb{C}^{2Q \times 2Q},
\end{aligned}
\label{equ:Fsum1}
\end{equation}
where $p(\bz;\btheta)$ is the probability density function (pdf) of $\bz$ parameterized by the unknown vector $\btheta$. Due to uncorrelated errors, the log-likelihood $\ln \: p(\mathbf{z};\btheta)$ is additive, and given by
\begin{equation}
\ln\: p(\mathbf{z};\btheta) = \sum_{r=1}^{R} \sum_{i=1}^{I} \sum_{p=1}^{P}{\sum_{n=1}^{N}{\ln\: p(z_{r,i,p}[n];\btheta)}}.
\label{equ:lpsum1}
\end{equation}
Due to~\eqref{equ:lpsum1} the FIM in~\eqref{equ:Fsum1} is also additive and can be written as~\cite{chepuri2015sparsity}

\begin{equation}
\bF(\btheta) = \sum_{r=1}^R{\sum_{i=1}^I{\sum_{p=1}^P{\bF_{r,i,p}(\btheta)}}},
\label{equ:fsum1}
\end{equation}
where $\bF_{r,i,p}(\btheta)$ is the FIM of the $p$th pulse due to the transmitter-receiver pair $(i,r)$ for all the $N$ samples, i.e.,
\begin{equation}
\begin{aligned}
\bF_{r,i,p}(\btheta) &= \sum_{n=1}^N{\bF_{r,i,p,n}(\btheta)} \\&= \frac{4}{\sigma_e^2}\sum_{n=1}^N{\frac{\partial y_{r,i,p}[n]}{\partial \btheta}\frac{\partial y_{r,i,p}[n]}{\partial \btheta^H}}.
\label{equ:fsum2}
\end{aligned}
\end{equation}

One of the contributions of this paper is to introduce the CRLB for two targets as a better performance measure for antenna/pulse selection in comparison with the CRLB for a single target. The reasoning is based on the fact that, for two targets, the correlation between the signals echoed from the targets is taken into account and both the estimation accuracy (mainlobe width of the ambiguity function) and the robustness against ambiguities (the sidelobe level around the mainlobe of the ambiguity function) are accounted for in the cost function that is used to optimize the antenna/pulse placement. In contrast, in the single-target CRLB, only the estimation accuracy of one target is considered, which essentially makes the mainlobe width as narrow as possible, ignoring the occurrence of ambiguities due to the sidelobes (i.e., a high sidelobe level around the mainlobe might occur due to the nonuniform antenna/pulse placement). Due to these reasons, we employ the two-target CRLB as a performance measure in our optimization problems. In the following, the FIMs for the two-target case for all the $2Q$ unknown parameters are derived. 
\subsection*{Two-target CRLB}
In this scenario, two targets are considered in the region of interest and the CRLB is derived for these two targets. For $Q=2$, (\ref{equ:Qtarget}) simplifies to
\begin{equation}
\begin{aligned}
z_{r,i,p}[n] &= y_{r,i,p}[n] + e_{r,i,p}[n] \\
&= \alpha_1h(pT_P+nT_s;v_1)\beta_1(nT_s)\phi_{r,i}(u_1) \\&+ \alpha_2h(pT_P+nT_s;v_2)\beta_2(nT_s)\phi_{r,i}(u_2) \\&+ e_{r,i,p}[n].
\end{aligned}
\end{equation}
The partial derivative of the signal w.r.t. the unknowns is given by
\begin{equation}
\frac{\partial y_{r,i,p}[n]}{\partial \btheta} = \frac{j2\pi}{\lambda}
\begin{pmatrix} (d_i+d_r)y^{(1)}_{r,i,p}[n] \\ 2(nT_s+pT_P)y^{(1)}_{r,i,p}[n] \\ (d_i+d_r)y^{(2)}_{r,i,p}[n] \\ 2(nT_s+pT_P)y^{(2)}_{r,i,p}[n] \end{pmatrix},
\label{equ:partialderiv}
\end{equation}
where $y^{(1)}_{r,i,p}[n]$ and $y^{(2)}_{r,i,p}[n]$ are the noiseless signal terms due to the first and second target, respectively. This allows us to compute the Fisher information matrix as
\begin{equation}
\bF_{r,i,p,n} = \begin{bmatrix}\bJ_1^{(1)} &\bJ_2 \\ \bJ_2^* &\bJ_1^{(2)}  \end{bmatrix},
\label{equ:fpri2}
\end{equation}
where $\bJ_1^{(q)}$ is the single-target FIM for the $q$th target given by
\begin{equation}
\begin{aligned}
\bJ_1^{(q)} &= \frac{32\pi^2\alpha_q^2}{\lambda^2\sigma_e^2}
\\&\begin{bmatrix} &\frac{1}{2}(d_i+d_r)^2 &(d_i+d_r)(nT_s+pT_P) \\ &(d_i+d_r)(nT_s+pT_P) &2(nT_s+pT_P)^2  \end{bmatrix},
\end{aligned}
\label{equ:fpri}
\end{equation}
where $\frac{\alpha_q^2}{\sigma_e^2}$ is the signal to noise ratio (SNR). The expression (\ref{equ:fpri}) implies that the single-target FIM is independent of $\btheta$ and thus $\bF_{r,i,p}$ in (\ref{equ:fsum2}) and $\bF$ in (\ref{equ:fsum1}) are also independent of $\btheta$ for a single-target scenario. In addition, $\bJ_2\in \mathbb{C}^{2\times2}$ is the cross correlation between the signals of the two targets calculated as follows
\begin{equation}
\begin{aligned}
\bJ_2 &= \frac{32\pi^2\alpha_1\alpha_2^*}{\lambda^2\sigma_e^2}\\&\begin{bmatrix} \frac{1}{2}(d_i+d_r)^2 &(d_i+d_r)(nT_s+pT_P) \\ (d_i+d_r)(nT_s+pT_P) &2(nT_s+pT_P)^2  \end{bmatrix}\\&h(pT_p+nT_s;v_1-v_2)\phi_{r,i}(u_1-u_2).
\label{equ:f22}
\end{aligned}
\end{equation}
It is easy to see that the unknown parameters appear only in the cross correlation terms between the two targets. Moreover, the Fisher information matrix only depends on the difference between direction cosines and on the difference between the velocities. As will be seen later, we use this characteristic to reduce the search space. The final expression of the FIM is calculated from (\ref{equ:fsum1}) and (\ref{equ:fsum2}).

The calculated CRLB and Fisher information matrices are useful when all the unknown parameters in $\btheta$ have the same units. However, in this paper we have parameters with different units such as direction cosine and velocity (i.e., cosine of radians and $m/s$). Moreover, the desired estimation accuracy for the two targets might be different as well. In fact, if the estimation error of one of the parameters is much higher than that of the others, then that parameter would play the dominant role in the optimization problem and the selection would be based on that parameter solely. As a result, the final design would not be satisfactory in terms of the other parameters' estimation accuracy. Thus, to make a balance among the parameters, we introduce compensation weights and modify the CRLB matrix as
\begin{equation}
\bC'(\btheta) = \diag({\boldsymbol{\gamma}})\bC(\btheta)\diag({\boldsymbol{\gamma}}),
\end{equation}
where $\gamma_i^2$ is the known compensation weight for the $i$th unknown parameter which depends on the application and $\boldsymbol{\gamma}=[\gamma_1,...,\gamma_{2Q}]^T$. Similar to the modified CRLB matrix, we can define the modified Fisher matrix as
\begin{equation}
\bF'(\btheta)= \diag^{-1}\{\boldsymbol{\gamma}\}\bF(\btheta)\diag^{-1}\{\boldsymbol{\gamma}\}.
\end{equation}

As seen in (\ref{equ:fpri2}) and (\ref{equ:f22}), the two-target CRLB (unlike the single-target CRLB) is a function of the unknown parameters. Therefore, while optimizing the CRLB, in order to keep the optimization problem tractable, we grid the region of interest into a discrete set of points for which we can evaluate the CRLB. Since the two-target CRLB only depends on the difference between the direction cosine and velocity [cf. (\ref{equ:f22})], we only grid these differences in the region of interest resulting in the set $\mathcal{D}=\{\delta\btheta_1,...\delta\btheta_D\}$, where $\delta\btheta_d$ denotes the $d$th difference between the two targets' parameters. Hence, a 1-D scan of the difference of these parameters suffices to obtain all feasible two-target CRLB matrices. Since the CRLB is a matrix, in the next section, we introduce scalar measures of the CRLB as optimization criteria to design the antenna/pulse placement that should be optimized over these grid points.

\section{Problem formulation}
In this work, on one hand, we want to reduce the sensing cost, i.e., reduction in the number of transmitters, pulses, and receivers, while guaranteeing a desired estimation error. The $R$ receivers next to the $I$ transmitters and $P$ pulses of each transmitter are the parameters that affect both the estimation quality and the sensing cost (hardware and computational complexity). Therefore, the selection problem might be posed in the following two ways. In the first problem, we minimize the sensing cost with a constraint on the estimation error. In the second problem, we minimize the estimation error with a constraint on the sensing cost. Since we know how many antennas are available, we focus on the second problem. The other case can be tackled in a similar way when the desired estimation error is known. 

We model the sensing framework by introducing the following sets: $(i)$ the set of selected transmitter-pulses $\mathcal{A}\subseteq\mathcal{P}$, where $\mathcal{P}=\{a_{1,1},a_{1,2},...,a_{I,P}\}$ is the set of all the $I P$ transmitter-pulses, and $(ii)$ the set of selected receivers $\mathcal{B}\subseteq\mathcal{R}$, where $\mathcal{R}=\{b_1,...,b_R\}$ is the set of all the $R$ receivers. In addition, we further introduce the transmitter-pulse selection matrix and the receiver selection vector for easier notation. That is, $\bA\in\mathcal{M\{\mathcal{A}\}}$ is a transmitter-pulse selection matrix defined by the set of selected transmitter-pulses $\mathcal{A}$, with $\mathcal{M\{\mathcal{A}\}}$ defining the singleton:
\begin{equation}
\mathcal{M\{\mathcal{A}\}} = \{\bA|\bA\in\{0,1\}^{I\times P}; [\bA]_{i,p}=1 \iff a_{i,p}\in \mathcal{A}\}.
\end{equation}
Here, the $(i,p)$th entry of $\bA$, denoted by $[\bA]_{i,p}$, is equal to $1(0)$ if the $i$th transmitter transmits the $p$th pulse (or not). In a similar way, we introduce $\bb\in\mathcal{V\{\mathcal{B}\}}$ as the receiver selection vector defined by the set of selected receivers $\mathcal{B}$, with $\mathcal{V\{\mathcal{B}\}}$ defining the singleton:
\begin{equation}
\mathcal{V\{\mathcal{B}\}} = \{\bb|\bb\in\{0,1\}^R; [\bb]_r=1 \iff b_r \in \mathcal{B} \},
\end{equation}
where the $r$th entry of $\bb$, denoted by $[\bb]_{r}$, is equal to $1(0)$ if the $r$th receiver is (not) selected.

Note that these definitions are in the most general form with complete freedom to select any of the transmitters, pulses, or receivers. For specific purposes, which are discussed later on, one may consider only the receiver selection vector (i.e., by employing all the transmitters transmitting all the pulses), only the transmitter-pulse selection matrix (i.e., by employing all the receivers), or the selection vectors for receivers and transmitters (where we assume that each active transmitter would transmit all the pulses). As there is a one-to-one relation between the matrix (vector) and the set $\mathcal{A}$ ($\mathcal{B}$), from now on we employ them interchangeably.

Using the selection variables $\bA$ and $\bb$, the collected measurements can be written as follows
\begin{equation}
\begin{aligned}
z_{r,i,p}[n] = [\bb]_r  [\bA]_{i,p} (y_{r,i,p}[n] + e_{r,i,p}[n]),
\end{aligned}
\label{equ:select1}
\end{equation}
where depending on whether a transmitter-receiver-pulse is selected, the measurement will be collected. It is easy to show that the Fisher information matrix [cf.~\eqref{equ:fsum1}] will be modified based on (\ref{equ:select1}) as
\begin{equation}
\bF(\bA,\bb,\delta\btheta) = \sum_{r=1}^R{[\bb]_r\sum_{i=1}^I{\sum_{p=1}^P{[\bA]_{i,p}\bF_{r,i,p}(\delta\btheta)}}}.
\label{equ:sumfisher}
\end{equation}
Because the two-target FIM is used, the difference of the two targets' parameters is considered.
The most general form of the optimization problem can be mathematically formulated as

\begin{equation}
\begin{aligned}
& \underset{\mathcal{A}\subseteq\mathcal{P},\mathcal{B}\subseteq\mathcal{R}}{\min}
& & g(f(\bA,\bb,\delta\btheta),\mathcal{D}) \\
& \text{subject to}
&& \bA\in \mathcal{M\{\mathcal{A}\}}, |\mathcal{A}| \leq K_P, \\
&&& \bb \in \mathcal{V\{\mathcal{B}\}}, |\mathcal{B}| \leq K_R, 
\end{aligned}
\label{equ:firstmin}
\end{equation}
where $K_P$ and $K_R$ are the maximum number of transmitter-pulses and receivers, respectively. Here, $f(\bA,\bb,\delta\btheta)$ is a function of the estimation error at the grid point $\delta\btheta \in \mathcal{D}$, $g(\cdot)$ is a general composition of the function $f(\cdot)$ evaluated over all the grid points in $\mathcal{D}$, e.g., maximization or average of $f(\bA,\bb,\delta\btheta)$ for all $\delta\btheta \in \mathcal{D}$, the sets $\mathcal{A}$ and $\mathcal{B}$ represent the selected transmitters-pulses and receivers, respectively. To guarantee an estimation accuracy level over all the grid points, $g(\cdot)$ should be the max function. To guarantee an average accuracy level, $g(\cdot)$ can be defined as the average over $\mathcal{D}$. Since, (\ref{equ:firstmin}) is a combinatorial optimization problem and NP-hard in nature \cite{joshi2009sensor}, we use convex relaxation techniques to employ convex optimization and surrogate submodular functions to employ greedy optimization as two general approaches to solve this problem.

As convex optimization requires a convex cost function and convex constraints, we require a function $f(\bA,\bb,\delta\btheta)$ that is convex and that the non-convex sets $\mathcal{A}$ and $\mathcal{B}$ are relaxed to obtain convex constraints.
Both the maximum and expected value for $g(\cdot)$ could be employed for convex optimization. By reformulation in its epigraph form, we will use maximization for convex optimization which in general leads to a semidefinite program (SDP) that has a cubic computational complexity.

The other approach to solve this problem is to employ submodular optimization which has been shown useful to solve combinatorial optimization problems \cite{khuller1999budgeted,coutino2017near,contino2017near}. A set function $f:2^{|\mathcal{N}|} \rightarrow \mathbb{R}$ is called submodular, if and only if, for every $\mathcal{S}_1\subseteq \mathcal{S}_2 \subseteq \mathcal{N}$ and $u\in \mathcal{S}_2^{^\mathsf{c}}$, it shows the property of diminishing returns, i.e.,
\begin{equation}
f(\mathcal{S}_1\cup\{u\}) - f(\mathcal{S}_1) \geq f(\mathcal{S}_2\cup\{u\}) - f(\mathcal{S}_2).
\end{equation}
It is known that, if the function $f$ is nondecreasing, normalized and submodular, then by employing a conceptually simple greedy algorithm, which starts with an empty (full) set, and in iteration $i$, adds the best (removes the worst) element to (from) the set, to maximize the function (not minimize), it is possible to obtain an $1-1/e$ approximation of the optimum value of $\underset{\mathcal{S}\subseteq \mathcal{N}, |\mathcal{S}|\leq K}{\max} \;f(\mathcal{S})$ for some cardinality $K$ \cite{sviridenko2004note}. Thus, if $g(f(\bA,\bb,\btheta),\mathcal{D})$ satisfies this property, then we can use the greedy algorithm with near-optimality guarantees.

In essence, we could say that the advantage of the convex optimization approach is its higher freedom in terms of objective functions and constraints. On the other hand, submodular optimization generally leads to low computational methods which makes it appropriate for large-scale scenarios. In Section VI, both algorithms will be explained in more detail.

\begin{figure*}
\centering
\psfrag{5 pulses - 1 target}{\tiny{5 pulses - 1 target}}
\psfrag{5 pulses - 2 targets}{\tiny{5 pulses - 2 targets}}
\psfrag{8 pulses - 1 target}{\tiny{8 pulses - 1 target}}
\psfrag{8 pulses - 2 targets}{\tiny{8 pulses - 2 targets}}
\psfrag{5 pulses-1 target}{\tiny{5 pulses - 1 target}}
\psfrag{5 pulses-2 targets}{\tiny{5 pulses - 2 targets}}
\psfrag{8 pulses-1 target}{\tiny{8 pulses - 1 target}}
\psfrag{8 pulses-2 targets}{\tiny{8 pulses - 2 targets}}
\psfrag{Ambiguity function (dB)}{\scriptsize{Ambiguity function (dB)}}
\psfrag{Velocity (m/s)}{\scriptsize{Velocity (m/s)}}
\psfrag{tr(C)}{\scriptsize{$\tr(\bC)$}}
\psfrag{# transmit pulses}{\scriptsize{\# transmit pulses}}
\psfrag{E-optimality}{\tiny{E-optimality}}
\psfrag{D-optimality}{\tiny{D-optimality}}
\psfrag{A-optimality}{\tiny{A-optimality}}
\psfrag{MSE of MLE}{\tiny{MSE of MLE}}
\subfigure[] { \includegraphics[width=.3\textwidth]{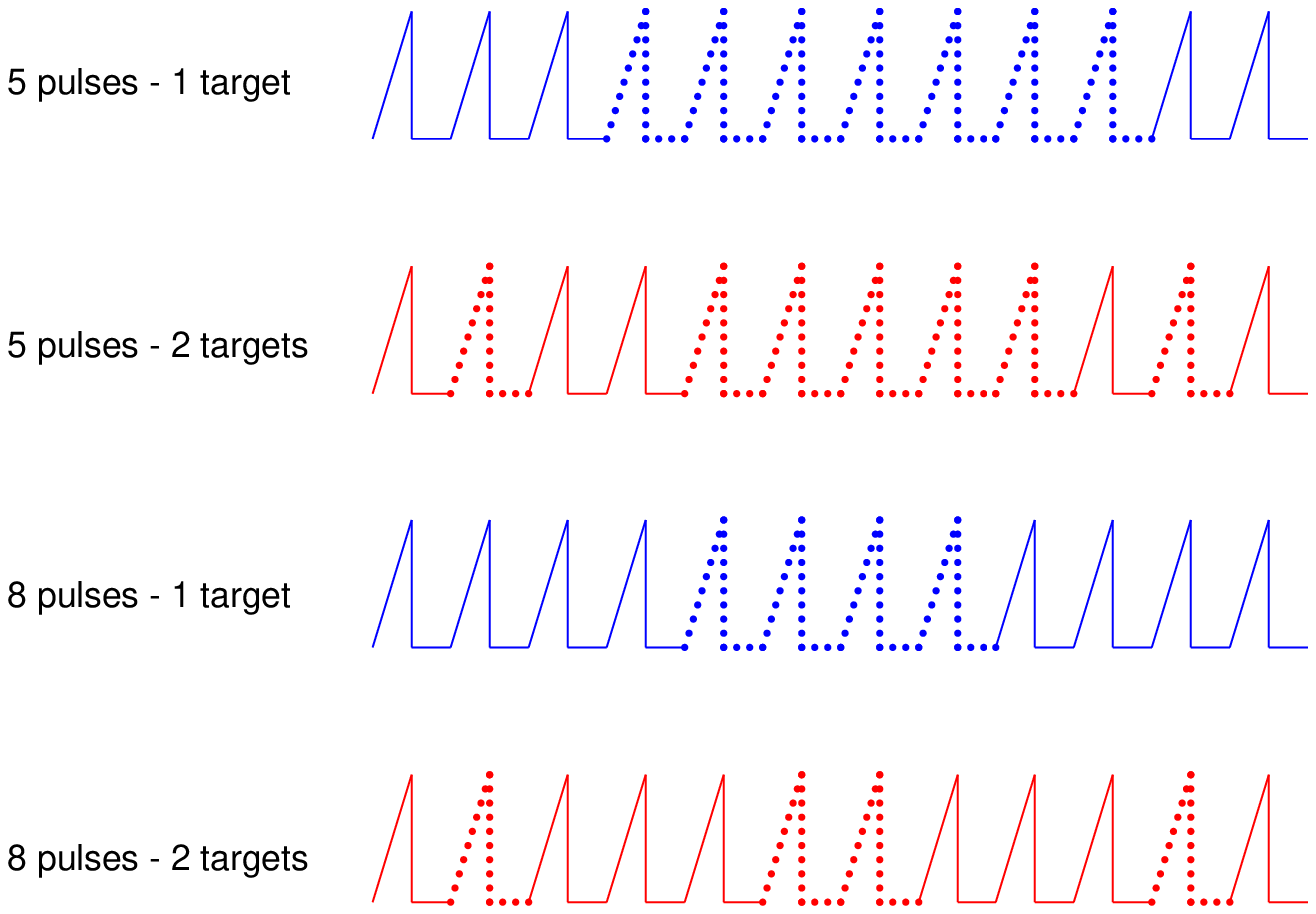}\psfrag{5}{$nd$}
\label{pulse_selection}} \quad
\subfigure[] {\includegraphics[width=.3\textwidth]{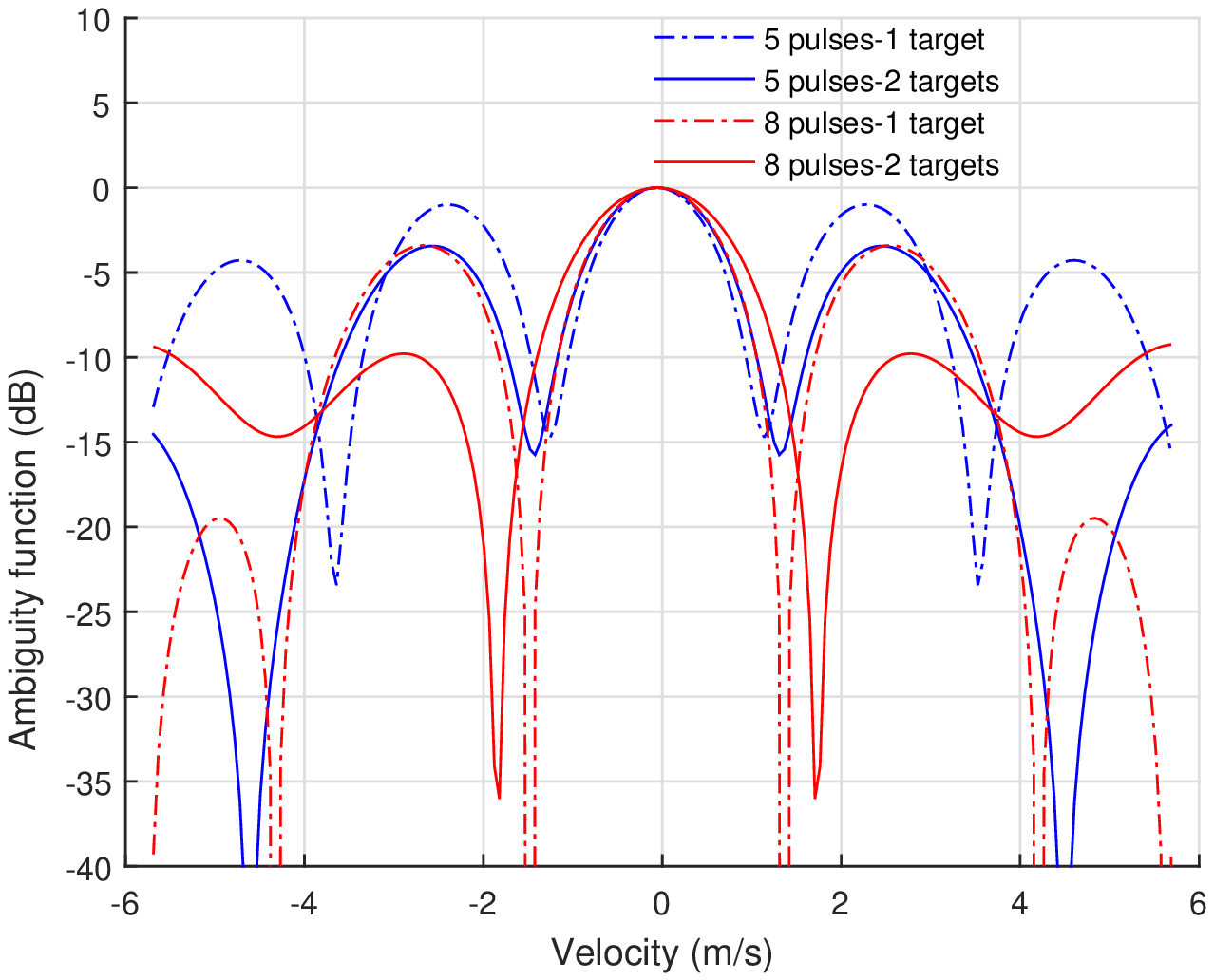}
\label{pulse_selection_patt}} \quad
\subfigure[] {\includegraphics[width=.3\textwidth]{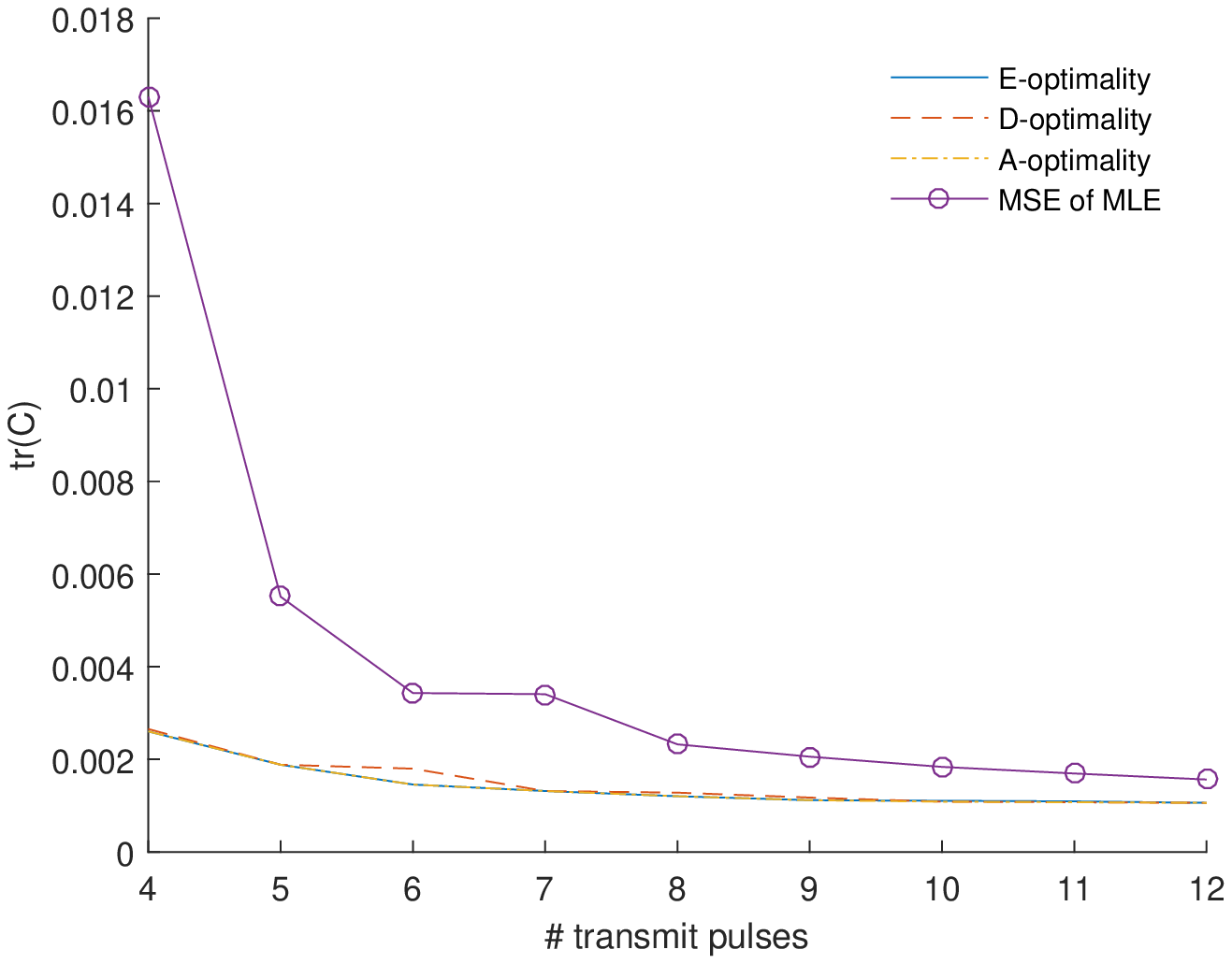}
\label{pulse_selection_MSE}} \\
\caption{\footnotesize{Single antenna pulse selection based on velocity estimation error for a total of $P=12$ pulses (a) selected pulses, (b) velocity ambiguity function, (c) scalar measures performance.}}
\label{fig2}
\end{figure*}

\subsection*{Scalar measures of the CRLB}
Since the CRLB is a matrix, it is not possible to employ it as an objective function for the optimization problem. Thus, in the following, scalar measures of the CRLB (or the FIM) that are employed in the proposed algorithms are introduced.

\begin{itemize}
	\item A-optimality: minimize the trace of the CRLB, i.e., $f(\bA,\bb,\delta\btheta) = \tr\{\bC(\bA,\bb,\delta\btheta)\}$.
	\item D-optimality: minimize the determinant of the CRLB, i.e., $f(\bA,\bb,\delta\btheta) = \log \det(\bF(\bA,\bb,\delta\btheta))$.
	\item E-optimality: minimize the maximum eigenvalue of the CRLB, i.e., $f(\bA,\bb,\delta\btheta) = \lambda_{max}\{\bC(\bA,\bb,\delta\btheta)\}$.
	\item Modified frame potential: the frame potential (FP) has been introduced in \cite{ranieri2014near} to measure orthogonality between vectors of a frame. 
	Due to the non-linearity of our model, we employ the first derivative of the measurements $\partial y_{r,i,p}[n]/\partial\boldsymbol\theta$ as defined in~\eqref{equ:partialderiv}. For each $i$,$r$, and $p$ which is selected, we have $N$ entries in the measurement matrix. Thus, the FP for our system model would be
	\begin{equation}
	{\rm{FP}}(\mathcal{S},\delta\btheta) = \sum_{y,y'\in\mathcal{Y\{\mathcal{S}\}}}\sum_{n=1}^N\left|\left<\partial y[n],\partial y'[n]\right>\right|^2,
	\label{equ:FP1}
	\end{equation}
	where $\partial y[n]$ is a simplified notation for $\partial y[n]/\partial\boldsymbol\theta$ for $y\in\mathcal{Y}\{\mathcal{S}\}$, $\mathcal{S}=\mathcal{A}\cup\mathcal{B}$ is the union set of transmitter pulses in set $\mathcal{A}$,
	and receivers in set $\mathcal{B}$,
	and $\mathcal{Y\{\mathcal{S}\}}=\{y_{r,i,p}[n], n\in\{1,\ldots,N\} \,| \,r\iff b_r \in \mathcal{B},(i,p) \iff a_{i,p} \in \mathcal{A}\}$ is the set of measurements due to the transmitter-pulses and receivers of $\mathcal{S}$ presented in \eqref{equ:Qtarget}. Even though the dependency with respect to $\delta\btheta$ is not explicitly stated in~\eqref{equ:FP1}, substituting \eqref{equ:partialderiv} in \eqref{equ:FP1}, it is straightforward to show that~\eqref{equ:FP1} is a function of the parameters difference vector, $\delta\btheta$.
	It has been shown that the FP performs the best under equal row norms and that the minimization of the FP and the MSE is related. However, in this problem, rows have different norms. On one hand, rows with lower norms are prioritized by the FP, but on the other hand, rows with higher norms contribute more to the estimation accuracy. Thus, we propose to normalize the rows and call the related FP the modified FP (MFP), which is given by
	\begin{equation}
	\begin{aligned}
	&\tilde{\rm{FP}}(\mathcal{S},\delta\btheta) =\\& \sum_{y,y'\in\mathcal{Y\{\mathcal{S}\}}}\sum_{n=1}^N\left|\frac{\left<\partial y[n],\partial y'[n]\right>}{\left<\partial y[n],\partial y[n]\right>\left<\partial y'[n],\partial y'[n]\right>}\right|^2.
	\end{aligned}
	\end{equation}
\end{itemize}

The above mentioned measures are employed as cost functions in different algorithms, which are presented in the following. It would be shown later that some are appropriate for convex optimization, while others are good for submodular optimization. It should be noted that each of these measures has some advantages and disadvantages. In other words, none of them are the best in general and based on the application and requirements, one may employ one or another. Finally, as mentioned before, the developed design approach can be adapted to widely-separated MIMO radars. Although there would be some minor changes in the signal model and the CRLB derivation, the overall idea would be the same and similar algorithms would be applicable.

\section{Two basic examples}
In this section, by explaining two simple examples, the idea behind this work is illustrated. In addition, some useful insights can be obtained from these examples. It should be mentioned that these are just small-scale examples for illustrating the general idea. As a result, an exhaustive search is used for solving the optimization problem~\eqref{equ:firstmin}. The proposed algorithms for large-scale problems are explained in Section VI.

\subsection{Single transmitter-receiver pair and multiple pulses}
In this example, we consider the problem of a single transmitter-receiver pair which is able to transmit $P$ identical pulses. The aim is to compare the estimation accuracy between employing all the pulses or just a few pulses after an appropriate selection. Omitting the transmitter and receiver indices, the measured signal for the $p$th pulse would be

\begin{equation}
\begin{aligned}
z_{p}[n] = [\ba]_{p} (y_{p}[n] + e_{p}[n]),
\end{aligned}
\end{equation}
where the matrix $\bA$ is now a column vector $\ba$, as a single transmitted is selected. In addition, as a single receiver is selected the vector $\bb$ is now omitted.
The optimization problem~\eqref{equ:firstmin} then simplifies to
\begin{equation}
\begin{aligned}
\underset{\ba}{\min} \;\;\underset{\delta\boldsymbol\theta \in \mathcal{D}}{\max}
&\;f(\ba,\delta\btheta) \\
\text{subject to}
& \;  ||\ba||_0 \leq K_P, \ba \in \{0,1\}^P,
\end{aligned}
\end{equation}
where $f(\ba,\delta\btheta)$ is one of the aforementioned measures, $\delta\btheta=v_1-v_2$ (since we are dealing with a single antenna pair there is no angle estimation, and the direction cosine is not considered as an unknown parameter), $\ba$ is the pulse selection vector, and $K_P$ is the constraint on the number of transmitted pulses. We solve this problem by performing an exhaustive search over all the possible combinations of pulses for both the single- and two-target CRLB criterion. In Figure~\ref{fig2}, the result of pulse selection on the velocity estimation error is represented where $P=12$ pulses in total are considered. In one case 5 and in the other case 8 pulses are selected. The result of pulse selection for these two cases for the single and two-target CRLB using A-optimality as the performance measure is shown in Figure \ref{pulse_selection}. It is clear that for the single-target CRLB case, the selection prioritizes the edges. However, for the two-target CRLB case, edge pulses are combined with intermediate pulses. This difference in pulse pattern causes the difference in the velocity ambiguity function which is depicted in Figure \ref{pulse_selection_patt}. It can be seen that, employing the two-target CRLB reduces the sidelobe level (especially for the sidelobes close to the mainlobe) at the price of a wider beamwidth. Finally, Figure \ref{pulse_selection_MSE} shows the trace of the two-target CRLB for different cost functions. Definitely, A-optimality is performing better than the others, because both the optimization cost function and the evaluation measure are the same. However, based on the plot, it turns out that all measures are performing similarly. In addition, the MSE of the maximum likelihood estimator (MLE) is also plotted for the optimal subset of pulses based on A-optimality, which shows the introduced measure is a good representative for the MSE. Note that although the optimization of the MSE of the estimator was the original aim, the MSE does not admit a closed form which makes it difficult to optimize. In contrast, calculating the surrogate measures we mentioned before is straightforward and based on Figure \ref{pulse_selection_MSE}, we observe that they are consistent with the MSE.

\subsection{MIMO radar and single pulse}
In this scenario, we are investigating another phenomenon, which is the effect of the antenna positions on the target angle estimation error. Since the number of pulses does not play a role in this example, a single pulse is considered for simplicity. The aim is to find the optimal antenna placement for a maximum angle estimation accuracy using different numbers of antennas. The optimization problem is as follows

\begin{equation}
\begin{aligned}
\underset{\ba,\bb}{\min} \; \underset{\delta\boldsymbol\theta \in \mathcal{D}}{\max}
& \; f(\ba,\bb,\delta\boldsymbol\theta) \\
\text{subject to}
& \; ||\ba||_0 \leq K_I, \ba \in \{0,1\}^I, \\
& \; ||\bb||_0 \leq K_R, \bb \in \{0,1\}^R,
\end{aligned}
\label{equ:pulse=1}
\end{equation}
where $f(\ba,\bb,\delta\btheta)$ is one of the different measures, $\delta\btheta = u_1-u_2$, $\ba$ and $\bb$ are the transmitter and receiver selection vectors, respectively, and $K_I$ and $K_R$ are the total number of selected transmitters and receivers, respectively. Similar to the previous example, the optimization is solved by performing an exhaustive search over all possible combinations of transmitters and receivers.

As an example, we perform the optimization for a total of 8 transmitters and 4 receivers and only consider A-optimality. Figure \ref{pulse=1} represents the result for two cases: 4 transmitters combined with 3 receivers and 6 transmitters combined with 2 receivers for the single and two-target CRLB. The selected antennas are depicted in Figure \ref{pulse=1_selection} for these four cases. As for the single antenna pulse selection example, the selected antennas for the single-target CRLB have a tendency to appear at the edges. However, for the two-target CRLB, antennas from both the edges and the middle of the array are selected. In addition, Figures \ref{pulse=1_patt43} and \ref{pulse=1_patt62} compare the beampatterns for the single and two-target CRLB. For both patterns, the sidelobe levels close to the mainlobe are reduced when the two-target CRLB is used in comparison with the examples obtained using the single-target CRLB. However, in Figure \ref{pulse=1_patt62}, higher sidelobes appear further away from the mainlobe. This effect is due to the fact that the sidelobes close to the mainlobe cause an ambiguity in distinguishing the two targets whereas the other sidelobes do not. Thus, the antenna selection focuses more on this issue. Note that it is possible to apply different weights to different u-coordinates in order to emphasize some specific regions in the beampattern.

\begin{figure*}
\centering
\psfrag{4 TX - 3 RX - 1 target}{\tiny{4 TX - 3 RX - 1 target}}
\psfrag{4 TX - 3 RX - 2 targets}{\tiny{4 TX - 3 RX - 2 targets}}
\psfrag{6 TX - 2 RX - 1 target}{\tiny{6 TX - 2 RX - 1 target}}
\psfrag{6 TX - 2 RX - 2 targets}{\tiny{4 TX - 3 RX - 2 targets}}
\psfrag{Transmitters}{\tiny{Transmitters}}
\psfrag{Receivers}{\tiny{Receivers}}
\psfrag{Normalized beam pattern (dB)}{\scriptsize{Normalized beam pattern (dB)}}
\psfrag{u}{\scriptsize{u}}
\psfrag{t=4r=3-1 target}{\tiny{$K_I$ = 4 \& $K_R$ = 3 - 1 target}}
\psfrag{t=4r=3-2 targets}{\tiny{$K_I$ = 4 \& $K_R$ = 3 - 2 targets}}
\psfrag{t=6r=2-1 target}{\tiny{$K_I$ = 6 \& $K_R$ = 2 - 1 target}}
\psfrag{t=6r=2-2 targets}{\tiny{$K_I$ = 6 \& $K_R$ = 2 - 2 targets}}
\subfigure[] {\includegraphics[width=.3\textwidth]{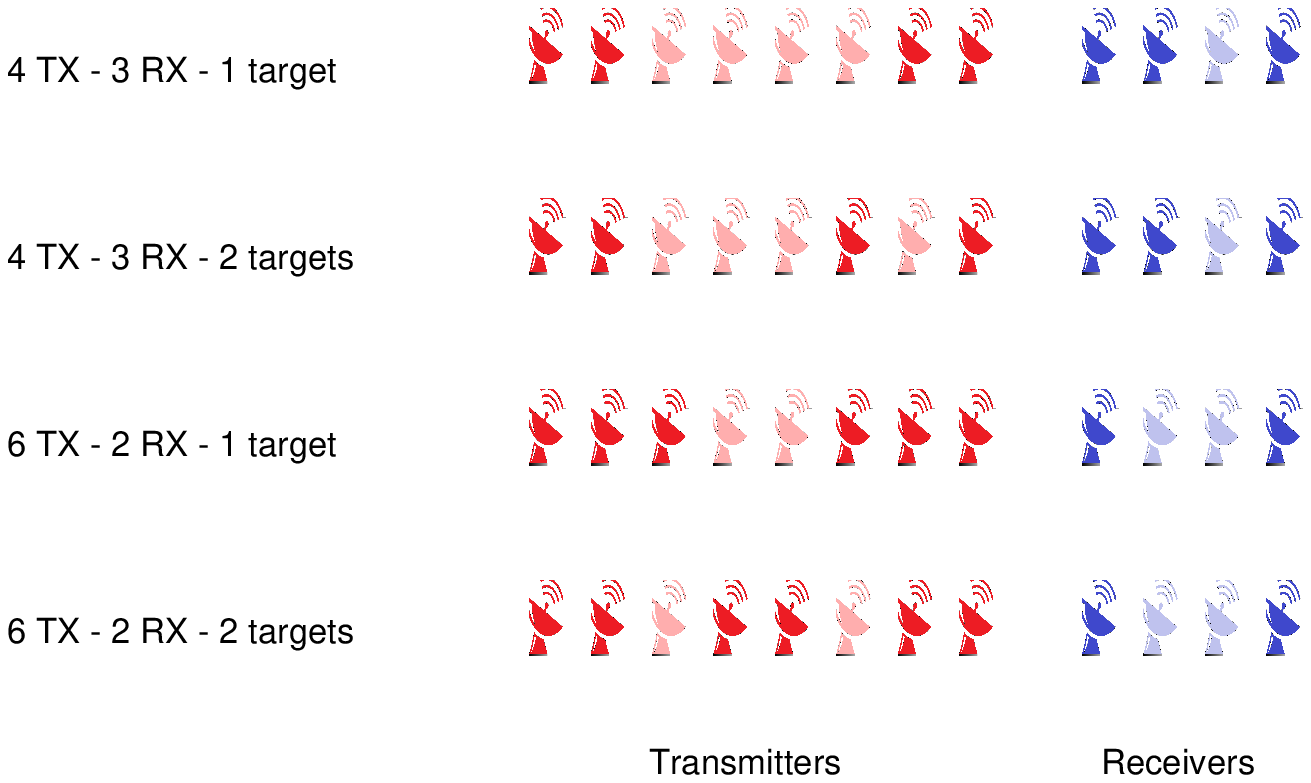}
\label{pulse=1_selection}} \quad
\subfigure[] {\includegraphics[width=.3\textwidth]{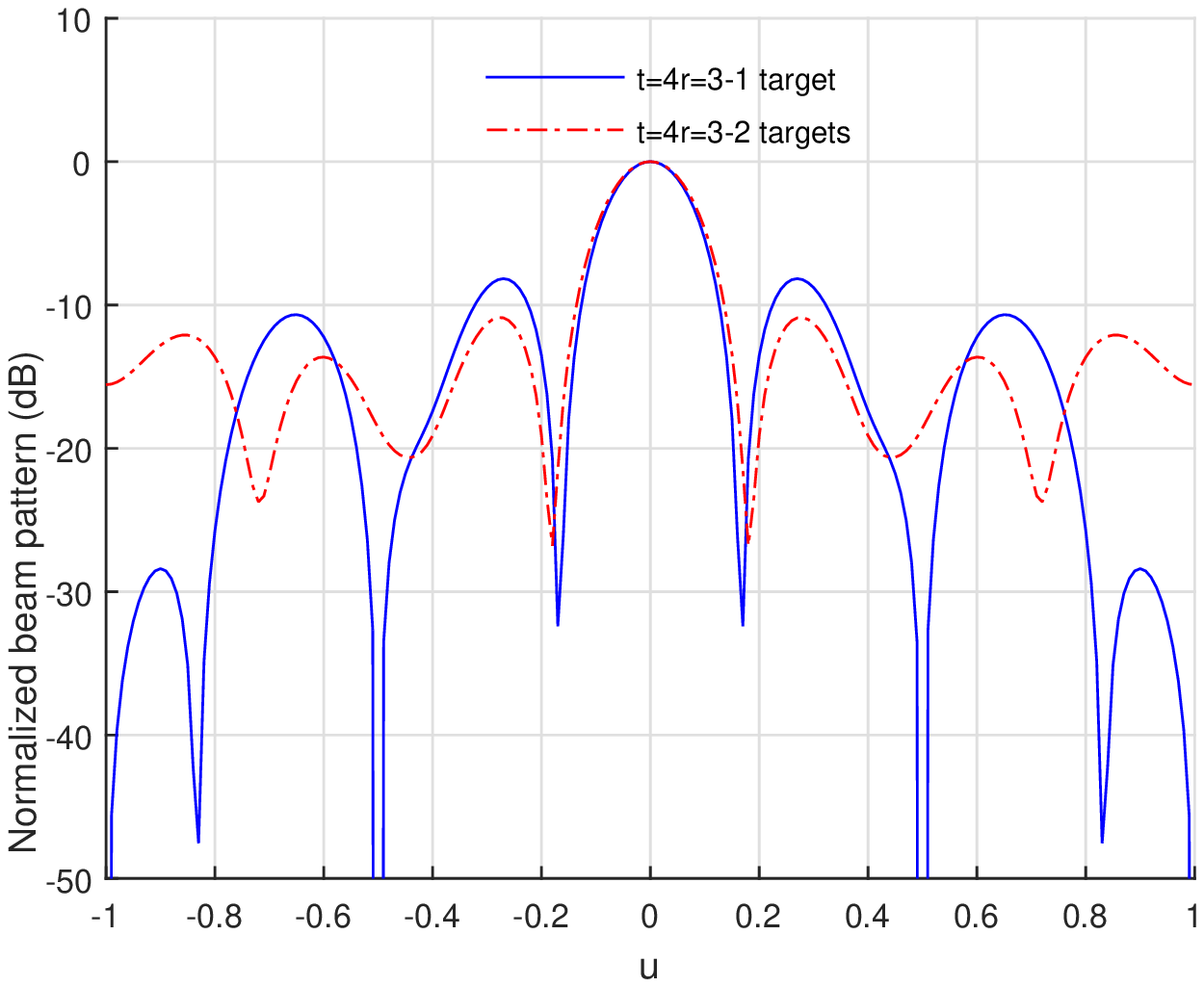}
\label{pulse=1_patt43}} \quad
\subfigure[] {\includegraphics[width=.3\textwidth]{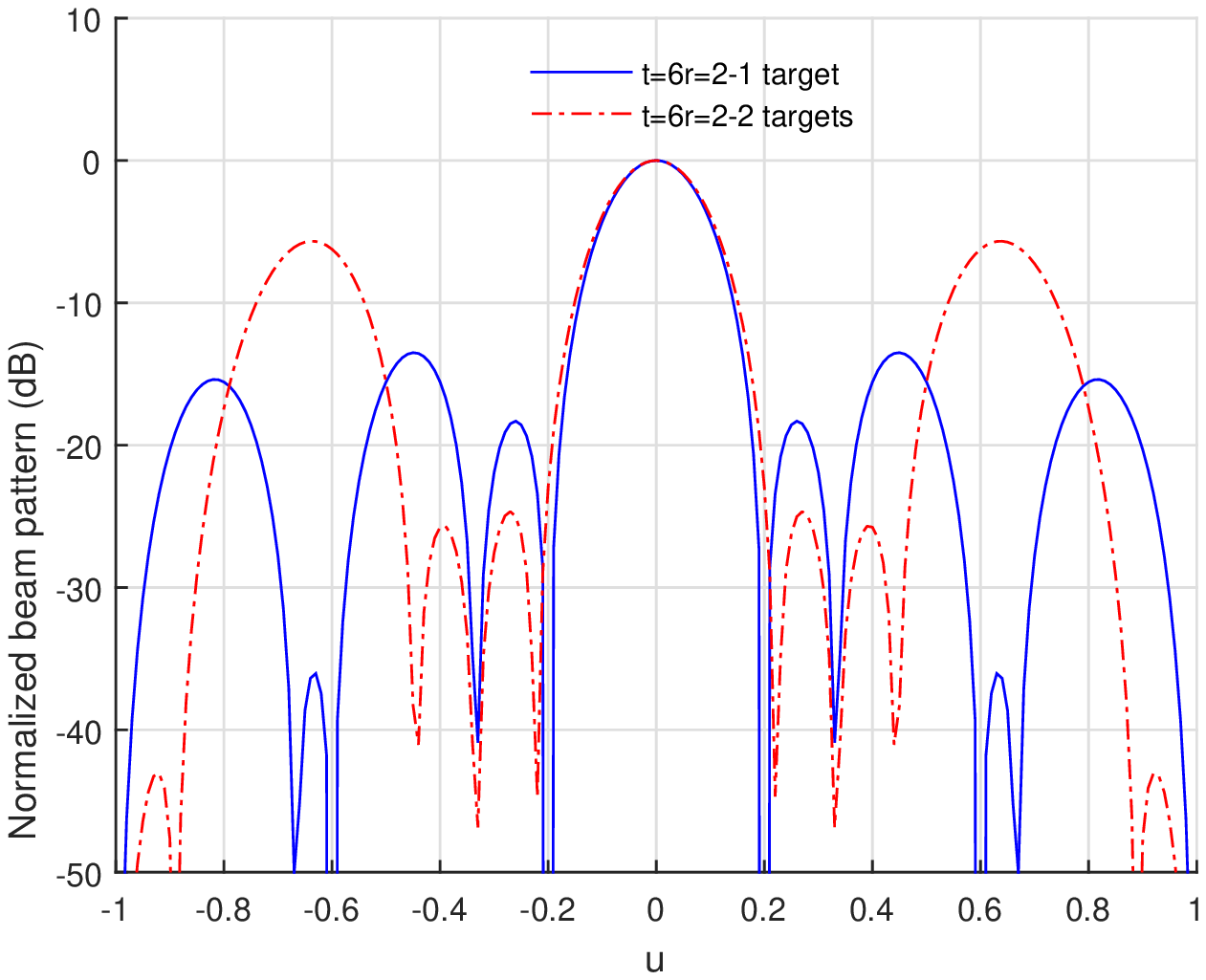}
\label{pulse=1_patt62}} \\
\caption{\footnotesize{MIMO radar antenna selection based on angle estimation error for a total of 8 transmitters and 4 receivers (a) Selected antennas, (b) beampatterns for $4$ transmitters and $3$ receivers, (c) beampatterns for $6$ transmitters and $2$ receivers.}}
\label{pulse=1}
\end{figure*}
\subsection{Discussion}
Based on the above two simple examples, it seems reasonable to seek the optimum sparse sensing scheme (both spatial and temporal) for different numbers of antennas and pulses and compare the estimation accuracy with full sensing. It may be possible to significantly reduce the number of samples at the price of only a small reduction in estimation accuracy. In the following sections, the general problem is stated, algorithms are proposed, and simulation results are presented.

\section{Transmitter-receiver selection}
Let us now study the most general case of transmitter-receiver-pulse selection. In other words, we would try to solve the original problem stated in (\ref{equ:firstmin}). It should be noted that, a transmitter is selected if and only if, it transmits at least one pulse. In the following, we would propose two general approaches to solve the problem: convex and submodular optimization.
\subsection{Convex optimization - E-optimality}
In this subsection, we try to solve the problem employing convex optimization. In principle, all scalar measures could be used since they are all convex, however we only consider E-optimality here because it is the easiest to formulate.
By restricting $g(\cdot)$ to be the maximum value and $f(\cdot)$ to be the maximum eigenvalue of the CRLB, and by relaxing the Boolean constraints in  (\ref{equ:firstmin}), the optimization problem can be written in the epigraph form as

\begin{equation}
\begin{aligned}
\underset{\bA,\bb,\gamma}{\max}
& \; \gamma \\
\text{subject to}
& \; \bF(\bA,\bb,\delta\btheta)  \succeq \gamma \bI_{4\times 4}, \forall \delta\btheta \in \mathcal{D} \\
& \; \sum_{i=1}^I{\sum_{p=1}^P{[\bA]_{i,p}}}\leq K_{P}, \\
& \; \sum_{r=1}^R{[\bb]_{r}}\leq K_{R}, \\
& \; 0\leq [\bA]_{i,p} \leq 1, 1\leq i \leq I, 1\leq p\leq P, \\
& \; 0 \leq [\bb]_{r} \leq 1, 1\leq r \leq R.
\end{aligned}
\label{equ:convex3}
\end{equation}
where $\bF(\bA,\bb,\delta\btheta)$ is the Fisher information matrix, $\bA$ and $\bb$ are the selection matrix and selection vector defined in (\ref{equ:sumfisher}), respectively, and $K_P$ and $K_R$ are the number of selected pulses and receivers, respectively.
Due to the presence of the products of unknowns (i.e., see \eqref{equ:sumfisher}), the optimization problem in (\ref{equ:convex3}) is not convex. Therefore, a convexifying process is introduced in several steps. First we define a pulse selection vector by vectorizing the selection matrix (i.e., $\ba=\vect(\bA)$). Then, we introduce the total selection vector $\bw$ by concatenating both the pulse and receiver selection vectors as
\begin{equation}
\bw = [{\ba}^T,{\bb}^T]^T.
\label{equ:concatenate}
\end{equation}
Finally, we introduce the total selection matrix as $\bW = \bw {\bw}^T$. Employing this new selection vector and matrix, the multiplication of the unknowns can be eliminated, and $\bW = \bw {\bw}^T$ is the only remaining non-convex term. Applying some standard convex relaxations on this term, the relaxed convex optimization problem can be stated as
\begin{equation}
\begin{aligned}
\underset{\bW,\bw,\gamma}{\max}
& \; \gamma \\
\text{subject to}
& \; \bF(\bW,\delta\btheta) \succeq \gamma \bI_{4\times 4}, \forall \delta\btheta \in \mathcal{D}, \\
& \; \begin{bmatrix} \bW  &\bw \\{\bw}^T   &1 \end{bmatrix} \succeq 0, \\
& \; [\bW]_{i,j} = [\bW]_{j,i}, 1\leq i,j \leq IP+R, \\
& \; [\bW]_{i,i} = [\bw]_i, 1\leq i \leq IP+R, \\
& \; \sum_{i=1}^{IP}{[\bw]_{i}}\leq K_{P}, \sum_{r=IP+1}^{IP+R}{[\bw]_{r}}\leq K_{R}, \\
& \; 0\leq [\bw]_{i} \leq 1, 1\leq i \leq I\times P+R,
\end{aligned}
\label{equ:convex4}
\end{equation}
where now the Fisher information matrix, $\bF(\bW,\delta\btheta)$, is reparametrized to be dependent in the introduced total selection matrix $\bW$. The optimization problem in (\ref{equ:convex4}) is a standard semidefinite programming problem in the inequality form which can be efficiently solved in polynomial time using interior-point methods. We can solve (\ref{equ:convex4}) with any of the off-the-shelf solvers. The solution of the relaxed optimization problem is used to compute the suboptimal Boolean solution for the selection problem. A straightforward technique that is often used is based on a simple sorting technique, in which the $K_{P}$ pulses corresponding to the largest values in $\bA$ and the $K_R$ receivers corresponding to the largest values in $\bb$ are selected as the transmitted pulses and receivers, respectively ($\bA$ and $\bb$ are obtained from the selection vector $\bw$ and considering (\ref{equ:concatenate})). However, randomized rounding is employed here which selects the antennas and pulses with a probability equal to the output of the convex problem. Details of randomized rounding are explained in \cite{chepuri2015sparsity}.

\subsection{Submodular optimization - MFP}

Although convex optimization is an efficient method, in this section, greedy submodular optimization is considered as a solution approach. The reason is the computational complexity which is much lower for greedy algorithms in comparison with convex optimization algorithms. This issue is especially important when dealing with large-scale scenarios.

Let us recall $\mathcal{P}$, the set of all transmitters-pulses, and $\mathcal{R}$, the set of all receivers. Furthermore, let us consider $\mathcal{A}\subseteq \mathcal{P}$ and $\mathcal{B} \subseteq \mathcal{R}$ as the set of selected pulses and receivers, respectively, and $\mathcal{S} = \mathcal{A}\cup\mathcal{B}$ as the union set of transmitter-pulses in $\mathcal{A}$ and receivers in $\mathcal{B}$. Finally, we define the ground set $\mathcal{U} = \mathcal{P}\cup\mathcal{R}$ as the union set of all the transmitter-pulses and receivers. 

Now, we introduce a set function $G : 2^{|\mathcal{U}|}\rightarrow\mathbb{R}_+$, defined over the subsets of the ground set $\mathcal{U}$, as the performance measure which is defined based on the modified frame potential as
\begin{equation}\label{eq:Gdef}
G(\mathcal{X}) = \tilde{\rm{FP}}(\mathcal{U}) - \tilde{\rm{FP}}(\mathcal{U}\setminus\mathcal{X}) \;\;\; \text{for } \mathcal{X}\subseteq\mathcal{U},
\end{equation}
where $\tilde{\rm{FP}}(\mathcal{U})$ and $\tilde{\rm{FP}}(\mathcal{X})$ are the MFPs due to the set of measurements $\mathcal{Y}(\mathcal{U})$ and $\mathcal{Y}(\mathcal{X})$, respectively. It is clear from the definition~\eqref{eq:Gdef} that aiming to maximize $G(\mathcal{S}^{\rm c})$, where $\mathcal{S}^{\rm c}$ is the complementary set of $\mathcal{S}$, i.e., $\mathcal{S}^{\rm c} = \mathcal{U}\setminus\mathcal{S}$, is tantamount to minimizing the MFP for the selected set of measurements $\mathcal{S}$. Therefore, it is possible to use~\eqref{eq:Gdef} as a performance metric to select the set of transmitter-pulses $\mathcal{S}$ by first identifying which elements should be \emph{discarded}. In the following, the next theorem guarantees the submodularity of the performance measure and thus gives near-optimal guarantees when the greedy algorithm is employed. 
\begin{theorem}
\label{theorem:tr-MFP}
For transmitter-receiver selection, $G: 2^{|\mathcal{U}|}\rightarrow \mathbb{R}_{+}$ is a normalized, monotone, submodular set function.
\end{theorem}

\begin{proof}
The proof is derived in Appendix\ref{proof 3}.
\end{proof}

The transmit pulse-receiver selection problem, using the performance metric defined in~\eqref{eq:Gdef} and the union set $\mathcal{S}$, can be now formally introduced as
\begin{equation}\label{eq:probG}
	\begin{array}{ll}
		\underset{\mathcal{S}^{\rm c}\,\subseteq\;\mathcal{U}}{\max} & G(\mathcal{S}^{\rm c})\\
		\text{subject to} & \mathcal{S}^{\rm c} \in \mathcal{I}_p(I P - K_P,R - K_R)
	\end{array},
\end{equation}
where $\mathcal{I}_p(I P - K_P,R - K_R)$ is a \emph{partition matroid}~\cite{schrijver2003combinatorial} whose independent sets are defined as
\begin{equation}\label{eq:defI}
	\begin{split} \mathcal{I}_p(I P - K_P,R - K_R) =  \{ \mathcal{X} : &|\mathcal{X}\cap \mathcal{P}| \leq I P - K_P,\\
			&|\mathcal{X}\cap \mathcal{R}| \leq R - K_R| \},
		\end{split}
\end{equation}   
leveraging the fact that $\{\mathcal{P},\mathcal{R}\}$ is a proper partition of $\mathcal{U}$. Due to the monotonicity of $G$ the maximum is achieved when the inequalities in the definition of the partition matroid are met with equality [cf.~\eqref{eq:defI}]. Therefore, the complementary set of the solution set of~\eqref{eq:probG} will meet the following properties:
\begin{equation}
	|\mathcal{S}\cap\mathcal{P}| = K_P,\;\; 	|\mathcal{S}\cap\mathcal{R}| = K_R,
\end{equation}
which are desired cardinality conditions for the set of selected transmitter-pulses. 
The following greedy algorithm is proposed for transmit pulse-receiver selection. At the starting point, all pulses and receivers are selected, i.e., $\mathcal{S} = \mathcal{P}\cup\mathcal{R}$. That is, we initialize the algorithm with $\mathcal{S}^{\rm c} = \emptyset$. Then, in each step, the greedy algorithm selects the element, either a receiver or transmit pulse, providing the highest cost function value and adds it to the set $\mathcal{S}^{\rm c}$. This procedure continues until the constraints are met. It should be noted that, if one of the constraints of the partition matroid is met with equality while the other is not, the proposed method continues adding elements (receivers or transmit pulses) until the desired cardinality is achieved. Fortunately, due to the structure of the ground set, and its partition, the independence oracle is easily implemented, i.e., routine for checking if a given set $\mathcal{S}$ is contained in a given matroid. Therefore, no overhead is incurred due to this procedure. The pseudocode of the algorithm is presented in {Algorithm 1}. The set returned by {Algorithm 1} achieves $1/2$ near-optimality guarantee~\cite{nemhauser1978analysis}. In the case that the matroid~\eqref{eq:defI} is substituted by a cardinality constraint on the set $\mathcal{S}^{\rm c}$, the greedy heuristic returns a $1 - 1/e$ near-optimal set. This situation can arise in instances when instead of having separated budget for transmit pulses and receivers, a joint budget is considered.

\begin{algorithm}[!b]\caption{Transmitter-pulse-receiver greedy selection based on MFP.}
\begin{algorithmic} 
\STATE \bf{Initialization:}
\STATE $\mathcal{S}^{\rm c} = \emptyset$
\STATE \bf{Greedy algorithm:}
\WHILE {$\mathcal{U}\neq \emptyset$}
\STATE $u^* = \underset{u\;\in\;\mathcal{U}}{\text{argmax}} \; G(\mathcal{S}^{\rm c}\cup\{u\}) $
\IF {$\mathcal{S}^{\rm c}\cup\{u^*\}\in\mathcal{I}_{p}$}
\STATE $\mathcal{S}^{\rm c} \leftarrow \mathcal{S}^{\rm c} \cup \{u^*\}$
\ENDIF
\STATE $\mathcal{U} \leftarrow \mathcal{U}\setminus\{u^*\}$
\ENDWHILE

\STATE $\mathcal{S} \leftarrow \mathcal{U}\setminus\mathcal{S}^{\rm c}$
\STATE $\mathcal{S}\cap\mathcal{P}$ and $\mathcal{S}\cap\mathcal{R}$ are the selected transmitters-pulses and receivers, respectively.
\end{algorithmic}
\end{algorithm}

\subsection{Fixed receivers - selection of pulses and transmitters}
The most general form of the optimization problem was studied in the previous section. In this scenario, we want to introduce the special case of fixed receivers. Proposing this special case is worthwhile for two reasons. First, in some applications we may have the freedom to only select pulses while the transmitters-receivers are fixed. The second reason is that this is a simpler version of the general case that helps to clarify part of the general case, i.e., in the procedure to solve the general problem, if the constraint is met for one of the parameters (i.e., pulses or receivers), the algorithm continues for the other parameter which is similar to this special case. In this case, while the receivers are considered to be fixed, we are trying to optimize the selection of pulses and transmitters to minimize the target's angle-velocity estimation error (the other case of fixed transmitter-pulses and the selection of receivers is similar). Since convex optimization based on E-optimality and submodular optimization based on the MFP for the general case have been already covered, we will not repeat these discussions here for this special case since they are similar and even simpler. However, we show here that in this case the log determinant is also a submodular function and it is possible to employ the greedy heuristic as an alternative optimization algorithm to solve the transmitter-pulse selection problem near optimally. It should be pointed out that the log determinant is not a submodular function for the general optimization problem and thus we only employ it as an objective function for this special case.

\subsubsection*{Submodular optimization - D-optimality}
In this case, we consider the log determinant set function be defined as
\begin{equation}\label{eq:lgFnc}
h(\mathcal{S}) = \begin{cases}
			0 & \Ss = \emptyset\\
			\log\det(\boldsymbol F_{\mathcal{S}}) & \text{otherwise}
	\end{cases},
\end{equation}
where $\boldsymbol F_{\mathcal{S}}$ is the Fisher information matrix [cf.~\eqref{equ:fsum1}] obtained by employing all the pulses in $\mathcal{S}\subseteq\mathcal{P}$. The set function~\eqref{eq:lgFnc} is employed as a performance measure (D-optimality). The greedy algorithm goes as follows. We start with all pulses and all transmitters (i.e., $\mathcal{S}=\mathcal{P}$). In each step, we remove the pulse that reduces the goal function the least. This procedure is repeated until we achieve the required number of pulses. The pseudocode of the greedy algorithm is presented in {Algorithm 2}. Submodularity of this cost function is proven in the following theorem, which ensures the $1-1/e$ performance bound of the greedy algorithm. 

\begin{theorem}
\label{theorem:fixed-logdet}
For pulse and transmitter selection, the set function $h : 2^{|\mathcal{P}|}\rightarrow\mathbb{R}_+$ [cf.~\eqref{eq:lgFnc}] is a normalized, monotone, submodular function.
\end{theorem}

\begin{proof}
The proof is derived in Appendix\ref{proof 1}.
\end{proof}

\begin{algorithm}[t]\caption{Transmit pulse greedy selection based on log determinant.}
\begin{algorithmic} 
\STATE \bf{Initialization:}
\STATE $\mathcal{S} = \mathcal{P};$
\STATE \bf{Greedy algorithm:}
\WHILE {$|\mathcal{S}|>K_P$}
\STATE $p = \underset{p \in \mathcal{S}}{\text{argmax}} \; h(\mathcal{S}\setminus\{p\})$ \textnormal{[cf.}~\eqref{eq:lgFnc}\textnormal{]}
\STATE $\mathcal{S} = \mathcal{S}\setminus \{p\};$
\ENDWHILE

\STATE $\mathcal{S}$  is the set of selected transmit pulses.
\end{algorithmic}
\end{algorithm}

\section {simulation results}


In this section, we study the performance of the proposed algorithms through numerical simulations. The simulations are performed for a radar using a 77-GHz frequency band with a 100 MHz bandwidth which is typically used for automotive radar systems \cite{wagner2013wide}. For the following simulations, we employed CVX to solve the convex optimization problems.

\subsection{Fixed receivers}
In this part, we test the performance of the proposed algorithms for the fixed receivers case. For the first scenario, three receivers, four transmitters, and four pulses are considered in total. All the receivers are assumed to be fixed. In Figure \ref{fixed_rec_MSE}, the results of the different algorithms are represented. The MSE for both angle and velocity estimation of all these three optimization algorithms in addition to the optimum MSE versus the number of pulses are presented in Figure \ref{fixed_rec_MSE}. This plot shows the performance of each algorithm and that their results are close to the optimum value. 



\begin{figure*}
\centering
\centering
\psfrag{MSE}{\scriptsize{MSE}}
\psfrag{# transmit pulses}{\scriptsize{\# transmit pulses}}
\psfrag{logdet-Greedy}{\tiny{logdet-Greedy}}
\psfrag{MFP-Greedy}{\tiny{MFP-Greedy}}
\psfrag{mineigen-Convex}{\tiny{mineigen-Convex}}
\psfrag{exhaustive search}{\tiny{exhaustive search}}
\subfigure[] {\includegraphics[width=.3\textwidth]{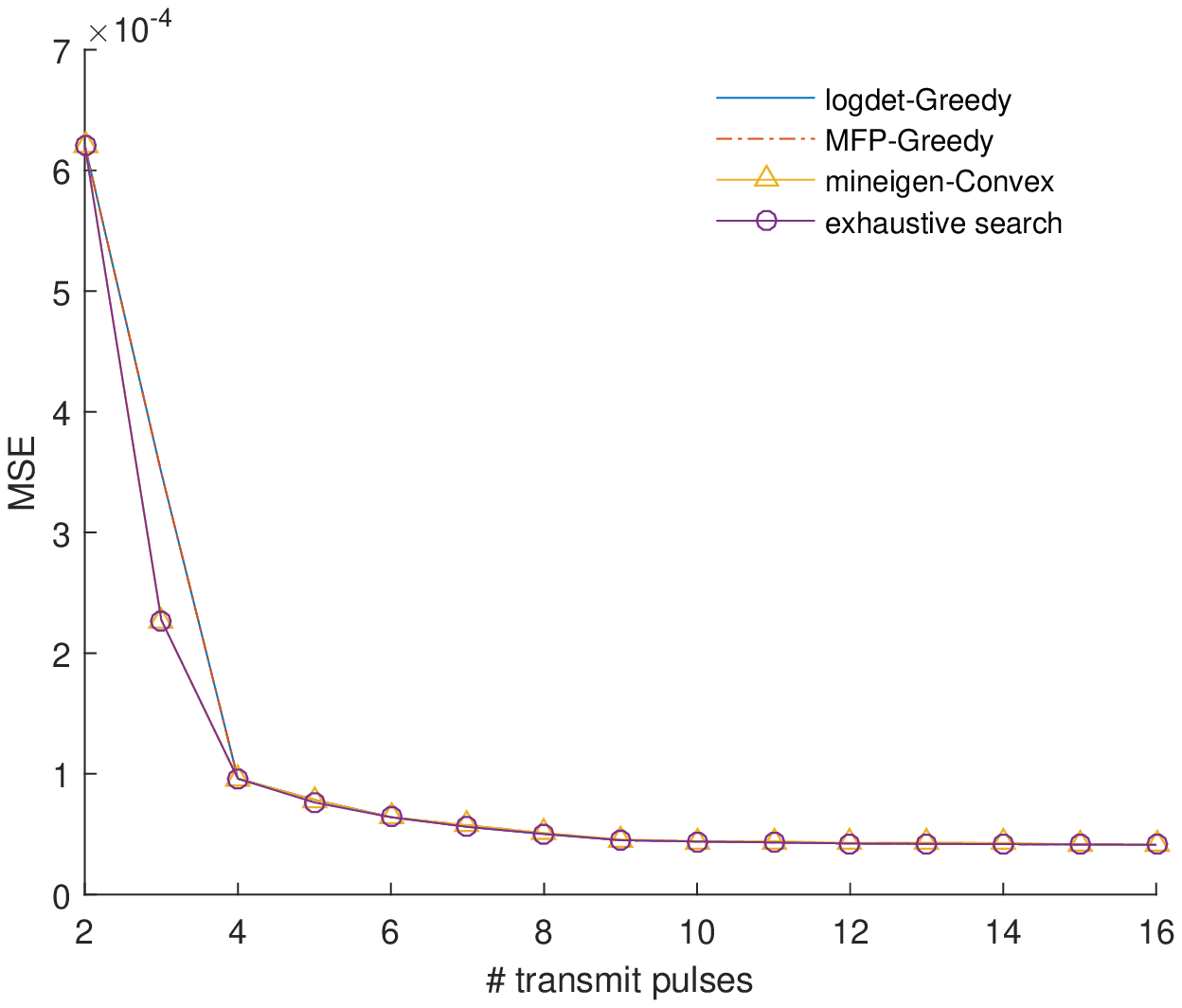}
\label{fixed_rec_MSE}} \quad
\subfigure[] {\includegraphics[width=.3\textwidth]{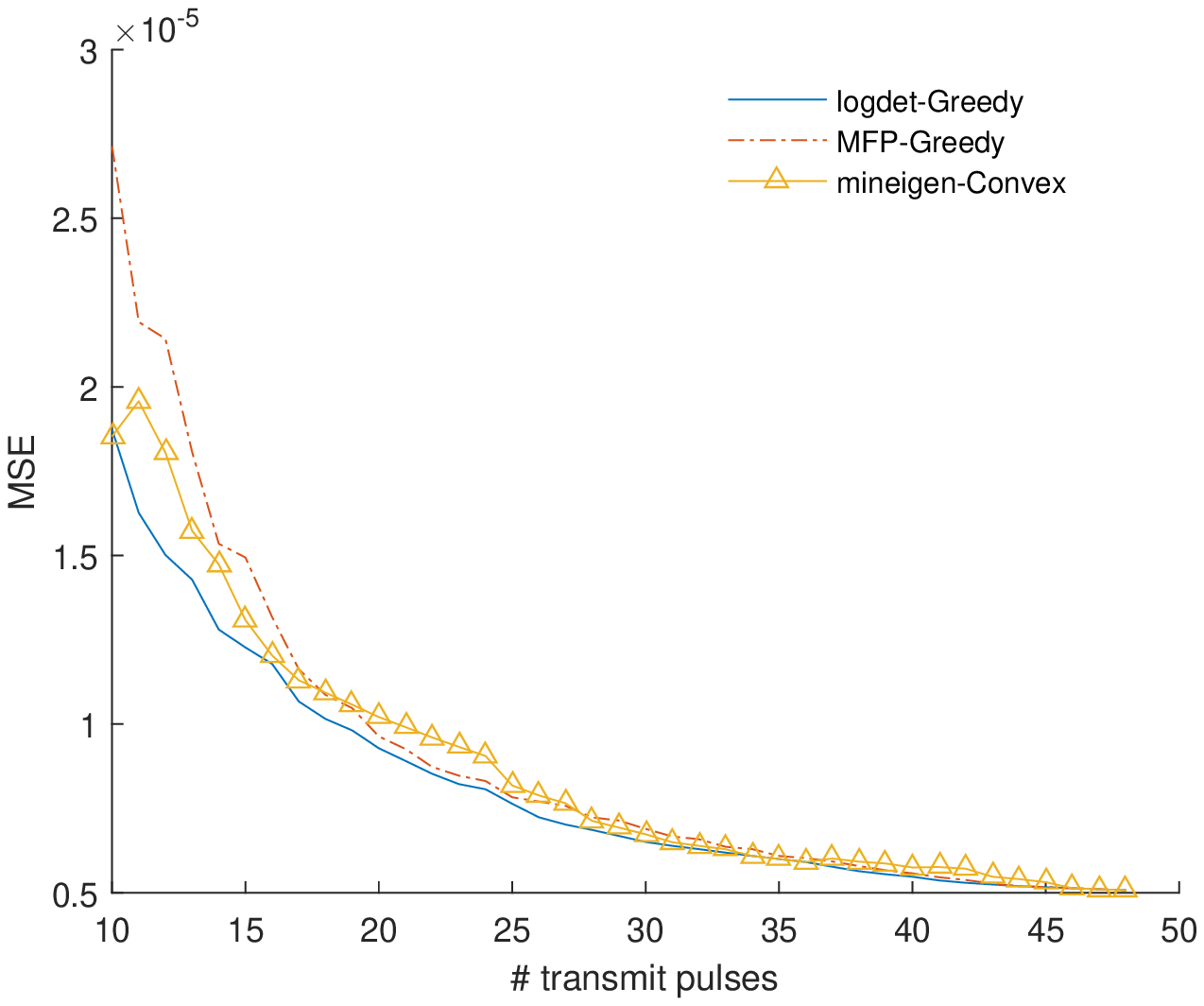}
\label{fixed_rec_48_MSE}} \quad
\psfrag{MSE}{\scriptsize{MSE}}
\psfrag{# transmit pulses}{\scriptsize{\# transmit pulses}}
\psfrag{MFP-Greedy,r=1}{\tiny{MFP-Greedy,\:\:\:\:\:\:\:\:  $K_R$=1}}
\psfrag{mineigen-Convex,r=1}{\tiny{mineigen-Convex, $K_R$=1}}
\psfrag{exhaustive search,r=1}{\tiny{exhaustive search, $K_R$=1}}
\psfrag{MFP-Greedy,r=3}{\tiny{MFP-Greedy,\:\:\:\:\:\:\:\:  $K_R$=3}}
\psfrag{mineigen-Convex,r=3}{\tiny{mineigen-Convex, $K_R$=3}}
\psfrag{exhaustive search,r=3}{\tiny{exhaustive search, $K_R$=3}}
\subfigure[] {\includegraphics[width=.3\textwidth]{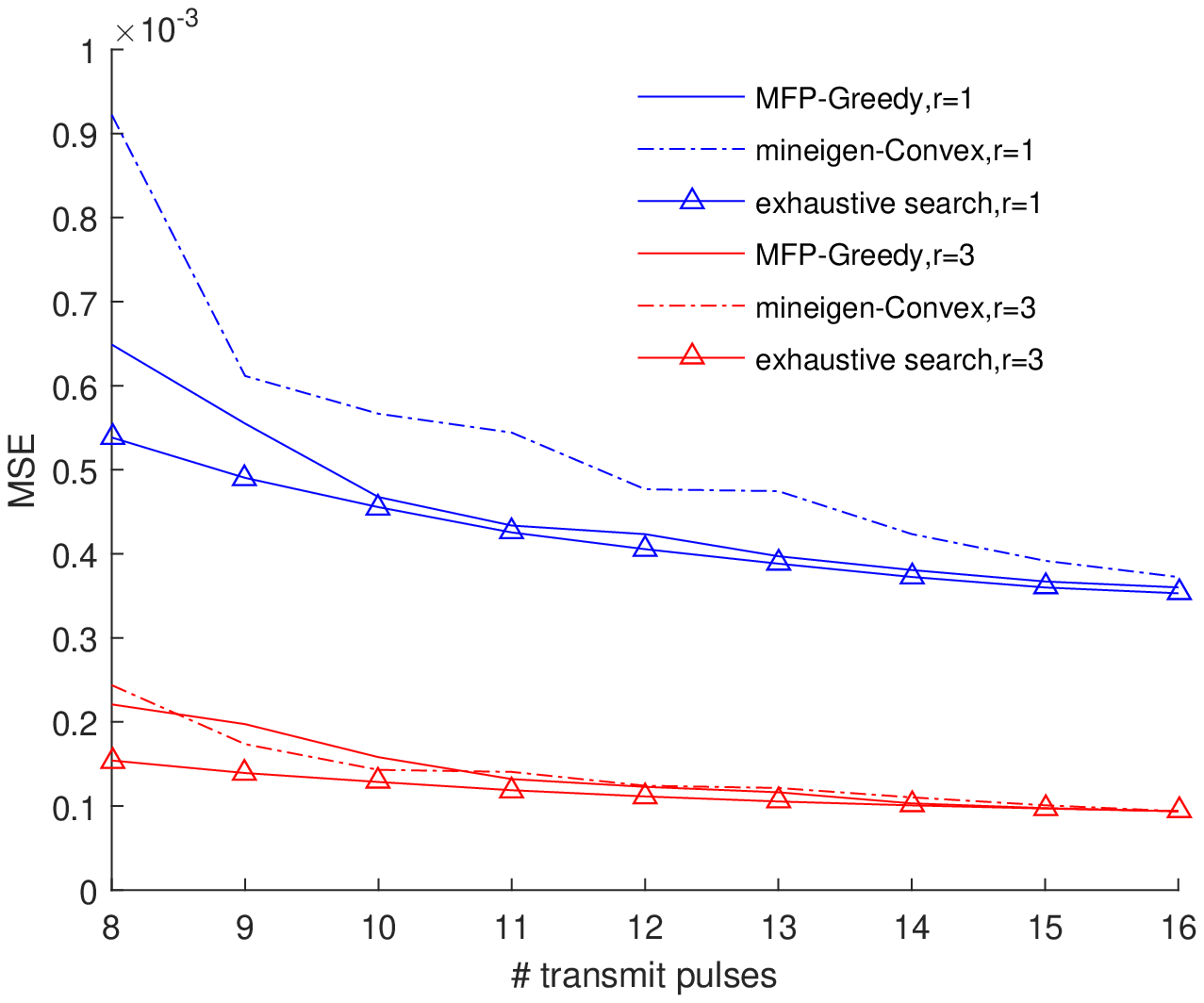}
\label{trp_MSE}} \\
\caption{\footnotesize{MSE versus number of transmitted pulses for (a) $16$ pulses in total, (b) $48$ pulses in total. (c) MSE versus number of transmit pulses for different approaches.}}
\label{mseFigures}
\end{figure*}


In another scenario, we consider two fixed receivers, six transmitters, and eight pulses in total. Figure \ref{fixed_rec_48_MSE} depicts the MSE for both angle and velocity estimation of the three proposed algorithms versus the number of transmitted pulses. All of them have a very close performance in terms of the MSE. In addition, figure \ref{fixed_rec_48_AF_xy} shows the ambiguity function for the result obtained by the submodular algorithm for the MFP when 24 pulses are selected. Here, a low sidelobe level and narrow beamwidth for both the direction cosine and velocity is achieved. The set of selected pulses is presented in Figure \ref{fixed_rec_48_selection}. As it is shown in Figure \ref{fixed_rec_48_selection}, pulses are selected from all the transmitters. Although, there is a tendency of selecting pulses towards the edges, the selected set includes different pulse numbers.




\begin{figure}
\centering
\psfrag{u}{\tiny{u}}
\psfrag{normalized AF}{\tiny{Normalized AF (dB)}}
\psfrag{Velocity (m/s)}{\tiny{Velocity (m/s)}}
\subfigure[] {\includegraphics[width=.25\textwidth]{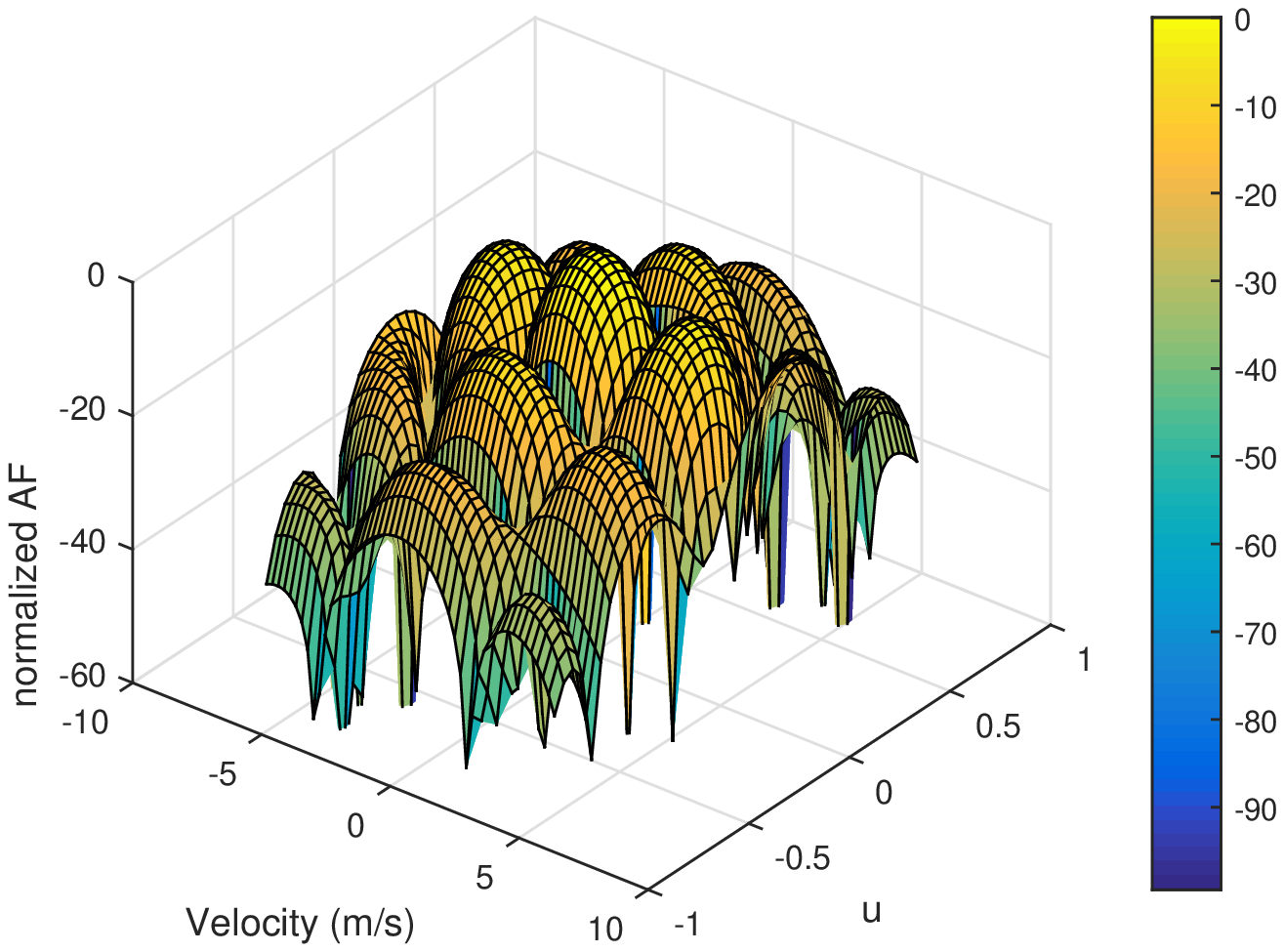}
\label{fixed_rec_48_AF_xy}}%
\subfigure[] {\includegraphics[width=.25\textwidth]{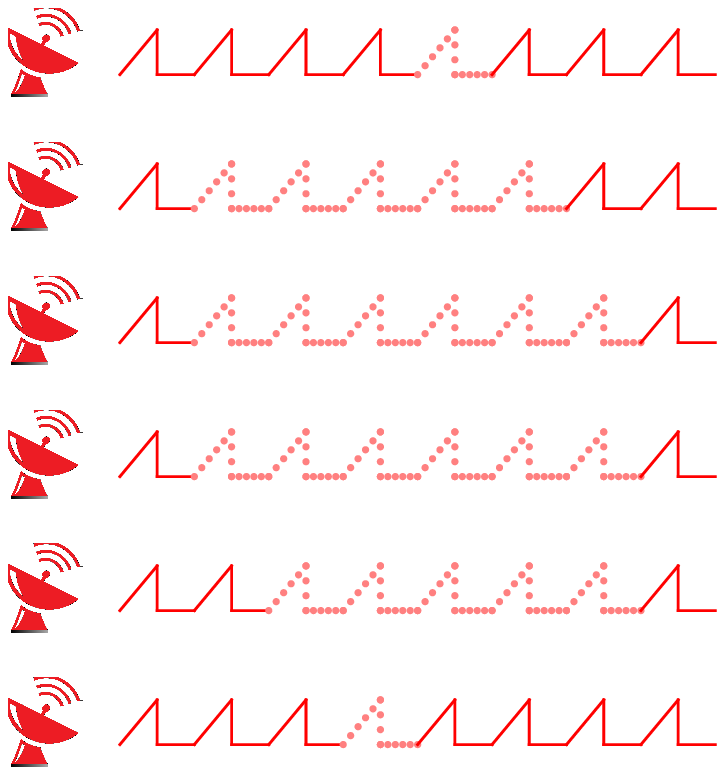}
\label{fixed_rec_48_selection}} \\
\caption{\footnotesize{MFP-submodular optimization result for 24 pulses (a) angle-velocity ambigutiy function, (b) selected transmitters-pulses. }}
\label{fixed_rec_48}
\end{figure}

\subsection{Transmitter-receiver-pulse selection}
Simulation results for the most general case of selecting transmitters-receivers-pulses is studied in this section. In total, we consider four receivers, four transmitters, and four pulses. Figure \ref{trp_MSE} presents the MSE for both angle and velocity estimation of the two optimization algorithms in addition to the optimum MSE. This plot shows again that the results are very close to the optimum value.  Note that the results are plotted for two different cases. In the first case, one out of four receivers is selected and in the second case, three out of four receivers are selected. It is clear that the MSE is lower for the last case. Moreover, Figure \ref{trp_12} presents the result of the submodular algorithm for the MFP when 12 pulses and 3 receivers are selected. The resulting ambiguity function and selected transmitters, receivers, and pulses are depicted in Figures \ref{trp_AF} and \ref{trp_selection}, respectively.

Finally, we consider a large-scale scenario with $20$ receivers, $20$ transmitters, and $10$ pulses in total (i.e., the total number of transmit pulses is $200$). It should be noted that due to the large number of parameters, the greedy algorithm is the only tractable optimization method. This is one of the advantages of submodular optimization over convex optimization. Figure \ref{trp_large_MSE} presents the MSE of the submodular algorithm for the MFP versus the numbers of selected transmit pulses for different number of selected receivers. As expected, the MSE decreases by increasing the number of transmit-pulses and receivers. However, it is shown that this improvement is saturated after a certain point. We could find some operating points in this figure such that by decreasing the performance slightly, a huge reduction in the number of transmit-pulses and receivers is achieved. For instance, the MSE for $80$ transmit-pulses and $8$ receivers is less than twice that of the full case, but with a much lower cost.


\begin{figure*}
\centering
\psfrag{normalized AF}{\scriptsize{Normalized AF (dB)}}
\psfrag{Velocity (m/s)}{\scriptsize{Velocity (m/s)}}
\psfrag{u}{\scriptsize{u}}
\subfigure[] {\includegraphics[width=.3\textwidth]{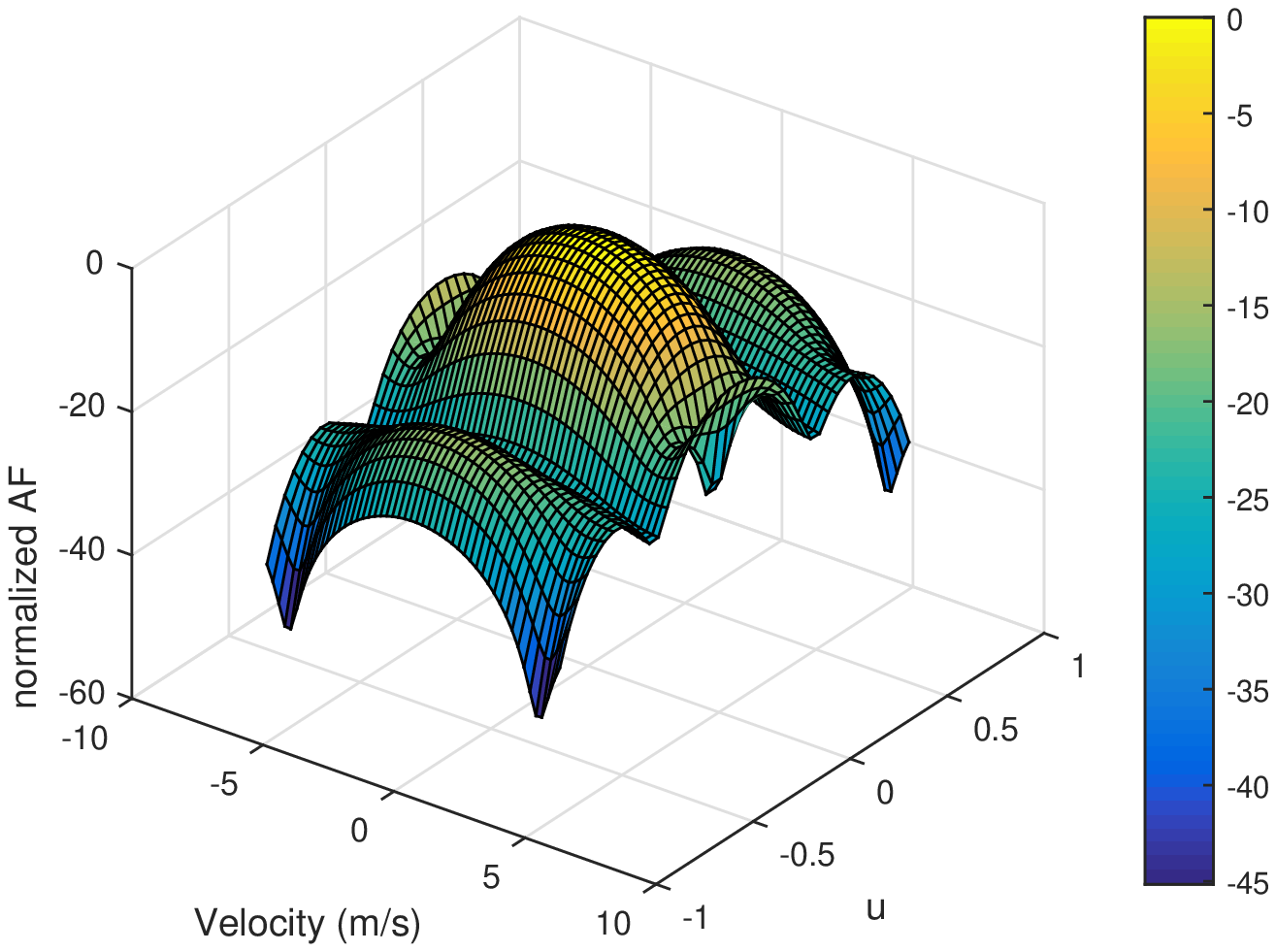}
\label{trp_AF}} \quad
\subfigure[] {\includegraphics[width=.3\textwidth]{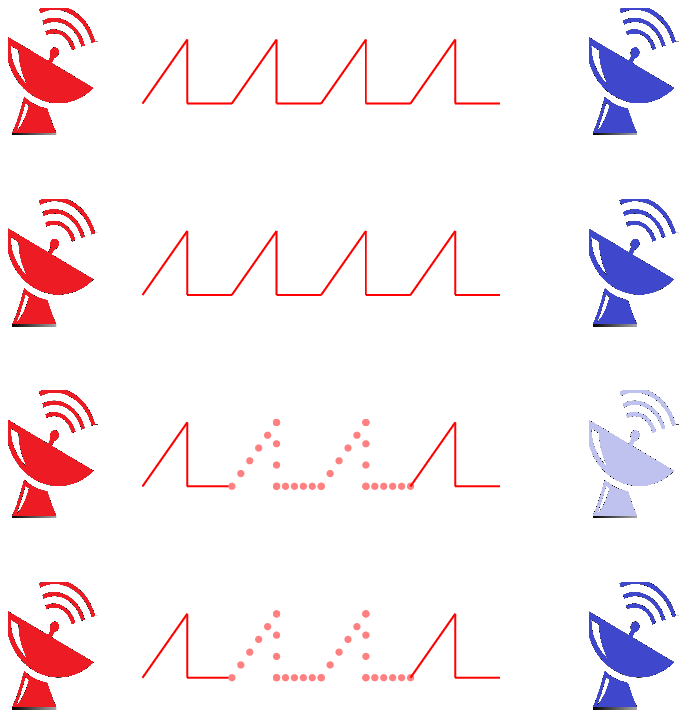}
\label{trp_selection}} \quad
\psfrag{MSE}{\scriptsize{MSE}}
\psfrag{# transmit pulses}{\scriptsize{\# transmit pulses}}
\psfrag{r=4}{\tiny{$K_R$ = 4}}
\psfrag{r=8}{\tiny{$K_R$ = 8}}
\psfrag{r=12}{\tiny{$K_R$ = 12}}
\psfrag{r=20}{\tiny{$K_R$ = 20}}
\subfigure[] {\includegraphics[width=.3\textwidth]{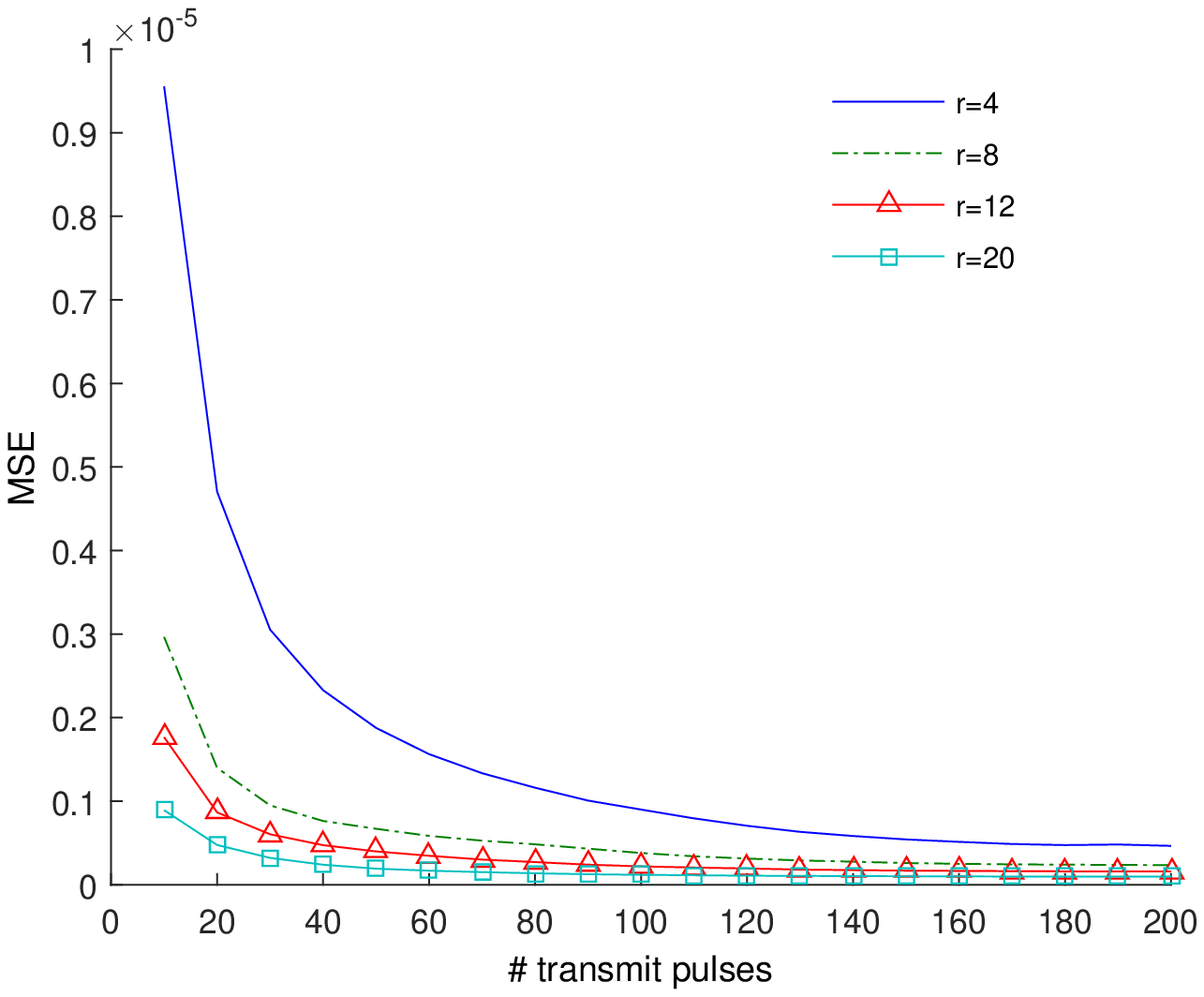}
\label{trp_large_MSE}} \\
\caption{\footnotesize{MFP-submodular optimization for 12 pulses and 3 receivers (a) angle-velocity ambiguity function, (b) selected transmitters-pulses-receivers. (c) MSE versus number of transmit pulses for a large scenario.}}
\label{trp_12}
\end{figure*}

\section{Conclusions}
In this paper, we presented algorithms to find the optimal set of antennas and pulses that achieves the minimum estimation error for different constraints on the number of antennas and pulses. It turned out that a significant reduction in the number of pulses and antennas with a small reduction in estimation accuracy is possible. Beside reducing hardware complexity (the number of antennas) and energy consumption (the number of pulses), the computational complexity is also reduced hugely due to the lower number of total samples. The one- and two-target CRLB for multiple antennas and multiple pulses were derived and it was shown that the two-target CRLB is a better measure for antennas and pulses selection optimization problem. Even though, several performance metrics were proposed, it should be stated that, there is no best solution for all problem instances and the appropriate performance metric should be selected based on the specific application. Convex and submodular optimization as the two different optimization approaches to antenna and pulse selection were introduced. It was shown that convex optimization provides more degrees of freedom in the optimization problem, i.e., it enables min-max optimization. On the other hand, the greedy submodular optimization obtains a near optimal solution with a low computational complexity which is desired especially in large-scale scenarios.

\appendices
\section*{Appendix}

\subsection{Proof of theorem \ref{theorem:tr-MFP}}
\label{proof 3}
\begin{proof}
First, we show that the function is \emph{normalized}. That is, $G(\emptyset) = 0$. This can be proved by noting
\begin{equation}
	G(\emptyset) = \tilde{\FP}(\Us) - \tilde{\FP}(\Us\setminus\emptyset) = \tilde{\FP}(\Us) - \tilde{\FP}(\Us) = 0.
\end{equation}
Now, we show the \emph{monotonicity} of $G$. Without loss of generality, we focus on the case that a new transmit pulse, $x\in\mathcal{P}$, is added as the proof for the other case (new receiver) can be constructed in a similar way. To show this, we require to show the following
\begin{equation}
	G(\Sc\cup\xs) - G(\Sc) \geq 0.
\end{equation}
First, we recall the definition~\eqref{eq:Gdef}, and expand the left-hand-side of the above inequality as
\begin{equation}\label{eq:step1}
	G(\Sc\cup\xs) - G(\Sc) = \tilde{\FP}(\Us\setminus\Sc) - \tilde{\FP}(\Us\setminus\{\Sc,x\}).
\end{equation}
Using the fact that $\Us\setminus\Sc = \Ss$ (by definition) we can rewrite~\eqref{eq:step1} as
\begin{equation}\label{eq:mono1}
	G(\Sc\cup\xs) - G(\Sc) = \tilde{\FP}(\hat{\Ss}\cup\xs) - \tilde{\FP}(\hat{\Ss}),
\end{equation}
where $\hat{\Ss} = \mathcal{S}\setminus \{x\}$. Substituting in~\eqref{eq:step1} the identity
\begin{equation}\label{eq:idenP}
	\tilde{\FP}(\hat{\Ss} \cup \xs) = \tilde{\FP}(\hat{\Ss}) + \tilde{\FP}(\xs \cup\Bs) + \tilde{\FP}(\hat{\Ss}, \xs\cup \Bs),
\end{equation}
where $\Bs = \hat{\Ss}\cap\mathcal{R}$, and we have defined the MFP for two sets (last term in the above expression) as
\begin{equation}
\begin{aligned}
&\tilde{\rm{FP}}(\mathcal{S}_1,\mathcal{S}_2) = \\&\sum_{y\in\mathcal{Y}(\mathcal{S}_1), y'\in\mathcal{Y}(\mathcal{S}_2)}\sum_{n=1}^N\left|\frac{\left<\partial y[n],\partial y'[n]\right>}{\left<\partial y[n],\partial y[n]\right>\left<\partial y'[n],\partial y'[n]\right>}\right|^2,
\end{aligned}
\end{equation}
we can show that
\begin{equation}
	\begin{split}
	G(\Sc\cup\xs) - G(\Sc) =  \tilde{\FP}&(\Bs\cup\xs) +\\
	& \tilde{\FP}(\hat{\Ss},\Bs\cup\xs) \geq 0\label{eq:mono2},
	\end{split}
\end{equation}
which proves the monotonicity of $G$.

Finally, we show the \emph{submodularity} of the set function. To do so, we restrict the proof to the general case, i.e., the elements involved in the proof are a transmit pulse, $x\in\mathcal{P}$, and a receiver, $y\in\mathcal{R}$. As for the case of monotonicity, this general proof can be particularized for the case in which both elements are of the same kind, i.e., two receivers, or two transmit pulses. 

To show submodularity we need to prove that
\begin{equation}
	G(\Sc\cup\xs) - G(\Sc) \geq G(\Sc\cup\{x,y\}) - G(\Sc\cup\{y\}).
\end{equation}
Expanding both sides of the inequality, we obtain
\begin{equation}
	\begin{split}
		\tilde{\FP}(\Us\setminus\Sc) - \tilde{\FP}(\Us\setminus\{\Sc,x\}) &\geq \tilde{\FP}(\Us\setminus\{\Sc,y\})\\
		& - \tilde{\FP}(\Us\setminus\{\Sc,x,y\}).
	\end{split}
\end{equation}
Using the identity $\tilde\Ss = \Us\setminus\{\Sc,x,y\}$ we can express the inequality as
\begin{equation}
	\begin{split}
		\tilde{\FP}(\tilde{\Ss}\cup\{x,y\}) - \tilde{\FP}(\tilde\Ss\cup\ys) &\geq \tilde{\FP}(\tilde\Ss\cup\xs)\\ &\;\;\;- \tilde{\FP}(\tilde\Ss).
	\end{split}
\end{equation}
Finally, using the identity~\eqref{eq:idenP} and the monotonicity of the MFP, we can show that
\begin{eqnarray}
	\tilde{\FP}(\tilde\Bs\cup\{x,y\}) + \tilde{\FP}(\tilde\Ss\cup\ys,\tilde\Bs\cup\{x,y\}) &\geq&\\
	\;\;\;\;\;\;\;\;\;\;\;\;\; \tilde{\FP}(\tilde\Bs\cup\xs) + \tilde{\FP}(\tilde\Ss,\tilde\Bs\cup\xs)\label{eq:defBt} &&\nonumber\\
	\begin{split} \tilde{\FP}(\tilde\Bs\cup\{x,y\}) - \tilde{\FP}(\tilde\Bs\cup\xs) \;+ \\ 
		\tilde{\FP}(\tilde\Ss\cup\ys,\tilde\Bs\cup\{x,y\}) - \tilde{\FP}(\tilde\Ss,\tilde\Bs\cup\xs)
	\end{split} &\geq& 0,
\end{eqnarray}
which proves the submodularity of the set function. In~\eqref{eq:defBt} we have defined $\tilde\Bs = \tilde\Ss\cap\mathcal{R}$ for readability. 
\end{proof}
\subsection{Proof of theorem \ref{theorem:fixed-logdet}}
\label{proof 1}
\begin{proof}
First, let us recall $\mathcal{P}$ as the set of pulses of all the transmitters and $\mathcal{S}\subset \mathcal{P}$ and $\bF_{\mathcal{S}}$ as the Fisher information matrix obtained by employing all the pulses in $\mathcal{S}$. Now, let $p_1,p_2\in \mathcal{P}\setminus \mathcal{S}$, then to prove submodularity we need to show

\begin{equation}
\begin{aligned}
&\log\det(\bF_{\mathcal{S}\cup \{p_1\}}) - \log\det(\bF_{\mathcal{S}}) \geq\\ &\log\det(\bF_{\mathcal{S}\cup \{p_1,p_2\}}) - \log\det(\bF_{\mathcal{S}\cup \{p_2\}}).
\end{aligned}
\label{equ:sub-dopt1}
\end{equation}
Noting that
\begin{equation}
\begin{aligned}
&\bF_{\mathcal{S}\cup \{p_1\}} = \bF_{\mathcal{S}} + \bF_{\{p_1\}}\\
&\bF_{\mathcal{S}\cup \{p_2\}} = \bF_{\mathcal{S}} + \bF_{\{p_2\}}\\
&\bF_{\mathcal{S}\cup \{p_1,p_2\}} = \bF_{\mathcal{S}} + \bF_{\{p_1\}} + \bF_{\{p_2\}},
\end{aligned}
\end{equation}
inequality (\ref{equ:sub-dopt1}) can be rewritten as

\begin{equation}
\begin{aligned}
&\log\det(\bF_{\mathcal{S}} + \bF_{\{p_1\}}) - \log\det(\bF_{\mathcal{S}}) \geq \\ &\log\det(\bF_{\mathcal{S}} + \bF_{\{p_1\}} + \bF_{\{p_2\}}) - \log\det(\bF_{\mathcal{S}} + \bF_{\{p_2\}}).
\label{equ:sublog}
\end{aligned}
\end{equation}
Thus, (\ref{equ:sublog}) implies
\begin{equation}
\frac{\det(\bF_{\mathcal{S}} + \bF_{\{p_1\}})\det(\bF_{\mathcal{S}} + \bF_{\{p_2\}})}{\det(\bF_{\mathcal{S}})\det(\bF_{\mathcal{S}} + \bF_{\{p_1\}} + \bF_{\{p_2\}})} \geq 1 .
\end{equation}
Considering $\bF_{\{p_1\}} = \bU\bV^T$ and employing the matrix determinant lemma in
\begin{equation}
\frac{\det(\bF_{\mathcal{S}})\det(\bI + \bV^T \bF_{\mathcal{S}}^{-1}\bU)\det(\bF_{\mathcal{S}} + \bF_{\{p_2\}})}{\det(\bF_{\mathcal{S}})\det(\bF_{\mathcal{S}} + \bF_{\{p_2\}})\det(\bI + \bV^T (\bF_{\mathcal{S}} + \bF_{\{p_2\}})^{-1}\bU)} \geq 1,
\end{equation}
leads to 
\begin{equation}
\frac{\det(\bI + \bV^T \bF_{\mathcal{S}}^{-1}\bU)}{\det(\bI + \bV^T (\bF_{\mathcal{S}} + \bF_{\{p_2\}})^{-1}\bU)} \geq 1,
\end{equation}
which is clear to be true since $\bF_{\{p_2\}}$ is a positive semi-definite matrix, i.e., $\bF_{\mathcal{S}}^{-1} \succeq (\bF_{\mathcal{S}} + \bF_{\{p_2\}})^{-1}$ as $\bF_{\mathcal{S}} \preceq \bF_{\mathcal{S}} + \bF_{\{p_2\}}$. Therefore, this function is submodular. In addition, monotonicity and normalization are clear from the definition.
\end{proof}


\ifCLASSOPTIONcaptionsoff
\newpage
\fi
\bibliographystyle{ieeetr}
\bibliography{ref}

\begin{thebibliography}{10}

\bibitem{TohidiAsilomar2017}
E.~Tohidi, H.~Behroozi, and G.~Leus, ``{Antenna and Pulse Selection for
  Colocated MIMO Radar},'' in {\em Proc. 51st Asilomar Conf. Signals Syst.
  Comput.}, in press.

\bibitem{5393291}
Q.~He, R.~S. Blum, H.~Godrich, and A.~M. Haimovich, ``Target velocity
  estimation and antenna placement for {MIMO} radar with widely separated
  antennas,'' {\em IEEE Journal of Selected Topics in Signal Processing},
  vol.~4, pp.~79--100, Feb 2010.

\bibitem{4350230}
J.~Li and P.~Stoica, ``{MIMO} radar with colocated antennas,'' {\em IEEE Signal
  Processing Magazine}, vol.~24, pp.~106--114, Sept 2007.

\bibitem{4408448}
A.~M. Haimovich, R.~S. Blum, and L.~J. Cimini, ``{MIMO} radar with widely
  separated antennas,'' {\em IEEE Signal Processing Magazine}, vol.~25, no.~1,
  pp.~116--129, 2008.

\bibitem{7126203}
H.~Xu, R.~S. Blum, J.~Wang, and J.~Yuan, ``Colocated {MIMO} radar waveform
  design for transmit beampattern formation,'' {\em IEEE Transactions on
  Aerospace and Electronic Systems}, vol.~51, pp.~1558--1568, April 2015.

\bibitem{6957532}
H.~Xu, J.~Wang, J.~Yuan, and X.~Shan, ``Colocated {MIMO} radar transmit
  beamspace design for randomly present target detection,'' {\em IEEE Signal
  Processing Letters}, vol.~22, pp.~828--832, July 2015.

\bibitem{6650099}
D.~S. Kalogerias and A.~P. Petropulu, ``Matrix completion in colocated {MIMO}
  radar: Recoverability, bounds amp; theoretical guarantees,'' {\em IEEE
  Transactions on Signal Processing}, vol.~62, pp.~309--321, Jan 2014.

\bibitem{6132420}
C.~Ma, T.~S. Yeo, C.~S. Tan, J.~Y. Li, and Y.~Shang, ``Three-dimensional
  imaging using colocated {MIMO} radar and {ISAR} technique,'' {\em IEEE
  Transactions on Geoscience and Remote Sensing}, vol.~50, pp.~3189--3201, Aug
  2012.

\bibitem{5752847}
C.~Ma, T.~S. Yeo, C.~S. Tan, and Z.~Liu, ``Three-dimensional imaging of targets
  using colocated {MIMO} radar,'' {\em IEEE Transactions on Geoscience and
  Remote Sensing}, vol.~49, pp.~3009--3021, Aug 2011.

\bibitem{5672411}
R.~Boyer, ``Performance bounds and angular resolution limit for the moving
  colocated {MIMO} radar,'' {\em IEEE Transactions on Signal Processing},
  vol.~59, pp.~1539--1552, April 2011.

\bibitem{6678708}
W.~Khan, I.~M. Qureshi, and K.~Sultan, ``Ambiguity function of phased-{MIMO}
  radar with colocated antennas and its properties,'' {\em IEEE Geoscience and
  Remote Sensing Letters}, vol.~11, pp.~1220--1224, July 2014.

\bibitem{7771508}
Y.~Kalkan and B.~Baykal, ``Frequency-based target localization methods for
  widely separated {MIMO} radar,'' {\em Radio Science}, vol.~49, pp.~53--67,
  Jan 2014.

\bibitem{6621849}
M.~Dianat, M.~R. Taban, J.~Dianat, and V.~Sedighi, ``Target localization using
  least squares estimation for {MIMO} radars with widely separated antennas,''
  {\em IEEE Transactions on Aerospace and Electronic Systems}, vol.~49,
  pp.~2730--2741, OCTOBER 2013.

\bibitem{5989873}
S.~Gogineni and A.~Nehorai, ``Target estimation using sparse modeling for
  distributed {MIMO} radar,'' {\em IEEE Transactions on Signal Processing},
  vol.~59, pp.~5315--5325, Nov 2011.

\bibitem{ender2010compressive}
J.~Ender, ``On compressive sensing applied to radar,'' {\em Signal Processing},
  vol.~90, no.~5, pp.~1402--1414, 2010.

\bibitem{godrich2010target}
H.~Godrich, A.~M. Haimovich, and R.~S. Blum, ``Target localization accuracy
  gain in {MIMO} radar-based systems,'' {\em IEEE Transactions on Information
  Theory}, vol.~56, no.~6, pp.~2783--2803, 2010.

\bibitem{he2016generalized}
Q.~He, J.~Hu, R.~S. Blum, and Y.~Wu, ``Generalized {Cram{\'e}r--Rao} bound for
  joint estimation of target position and velocity for active and passive radar
  networks,'' {\em IEEE Transactions on Signal Processing}, vol.~64, no.~8,
  pp.~2078--2089, 2016.

\bibitem{ai2015cramer}
Y.~Ai, W.~Yi, R.~S. Blum, and L.~Kong, ``{Cramer-Rao} lower bound for
  multitarget localization with noncoherent statistical {MIMO} radar,'' in {\em
  Radar Conference (RadarCon), 2015 IEEE}, pp.~1497--1502, IEEE, 2015.

\bibitem{7330290}
E.~Tohidi, M.~Radmard, S.~M. Karbasi, H.~Behroozi, and M.~M. Nayebi,
  ``{Compressive sensing in {MTI} processing},'' in {\em Proc. 3rd
  International Workshop on Compressed Sensing Theory and its Applications to
  Radar, Sonar and Remote Sensing (CoSeRa)}, pp.~189--193, June 2015.

\bibitem{godrich2012sensor}
H.~Godrich, A.~P. Petropulu, and H.~V. Poor, ``Sensor selection in distributed
  multiple-radar architectures for localization: A knapsack problem
  formulation,'' {\em IEEE Transactions on Signal Processing}, vol.~60, no.~1,
  pp.~247--260, 2012.

\bibitem{joshi2009sensor}
S.~Joshi and S.~Boyd, ``Sensor selection via convex optimization,'' {\em IEEE
  Transactions on Signal Processing}, vol.~57, no.~2, pp.~451--462, 2009.

\bibitem{roy2013sparsity}
V.~Roy, S.~P. Chepuri, and G.~Leus, ``Sparsity-enforcing sensor selection for
  {DOA} estimation,'' in {\em Computational Advances in Multi-Sensor Adaptive
  Processing (CAMSAP), 2013 IEEE 5th International Workshop on}, pp.~340--343,
  IEEE, 2013.

\bibitem{chepuri2015sparsity}
S.~P. Chepuri and G.~Leus, ``Sparsity-promoting sensor selection for non-linear
  measurement models,'' {\em IEEE Transactions on Signal Processing}, vol.~63,
  no.~3, pp.~684--698, 2015.

\bibitem{rao2015greedy}
S.~Rao, S.~P. Chepuri, and G.~Leus, ``Greedy sensor selection for non-linear
  models,'' in {\em Computational Advances in Multi-Sensor Adaptive Processing
  (CAMSAP), 2015 IEEE 6th International Workshop on}, pp.~241--244, IEEE, 2015.

\bibitem{greco2011cramer}
M.~S. Greco, P.~Stinco, F.~Gini, and A.~Farina, ``{Cram{\'e}r-Rao} bounds and
  selection of bistatic channels for multistatic radar systems,'' {\em IEEE
  Transactions on Aerospace and Electronic Systems}, vol.~47, no.~4,
  pp.~2934--2948, 2011.

\bibitem{ivashko2017radar}
I.~Ivashko, G.~Leus, and A.~Yarovoy, ``Radar network topology optimization for
  joint target position and velocity estimation,'' {\em Signal Processing},
  vol.~130, pp.~279--288, 2017.

\bibitem{rossi2014spatial}
M.~Rossi, A.~M. Haimovich, and Y.~C. Eldar, ``Spatial compressive sensing for
  {MIMO} radar,'' {\em IEEE Transactions on Signal Processing}, vol.~62, no.~2,
  pp.~419--430, 2014.

\bibitem{4487196}
D.~W. Bliss, K.~W. Forsythe, and C.~D. Richmond, ``{MIMO} radar: Joint array
  and waveform optimization,'' in {\em Conf. Rec. 41st Asilomar Conf. Signals
  Syst. Comput.}, pp.~207--211, Nov 2007.

\bibitem{skolnik2008radar}
M.~Skolnik, ``Radar handbook third edition,'' 2008.

\bibitem{wagner2013wide}
T.~Wagner, R.~Feger, and A.~Stelzer, ``Wide-band range-{Doppler} processing for
  {FMCW} systems,'' in {\em Radar Conference (EuRAD), 2013 European},
  pp.~160--163, IEEE, 2013.

\bibitem{chepuri2016sparse}
S.~Chepuri and G.~Leus, {\em Sparse Sensing for Statistical Inference}.
\newblock Now Publishers Incorporated, 2016.

\bibitem{6557996}
G.~Babur, O.~A. Krasnov, A.~Yarovoy, and P.~Aubry, ``Nearly orthogonal
  waveforms for {MIMO} {FMCW} radar,'' {\em IEEE Transactions on Aerospace and
  Electronic Systems}, vol.~49, pp.~1426--1437, July 2013.

\bibitem{ranieri2014near}
J.~Ranieri, A.~Chebira, and M.~Vetterli, ``Near-optimal sensor placement for
  linear inverse problems,'' {\em IEEE Transactions on signal processing},
  vol.~62, no.~5, pp.~1135--1146, 2014.

\bibitem{Kay2008statistical}
S.~M. Kay, ``Fundamentals of statistical signal processing: Estimation
  theory,'' 1993.

\bibitem{khuller1999budgeted}
S.~Khuller, A.~Moss, and J.~S. Naor, ``The budgeted maximum coverage problem,''
  {\em Information Processing Letters}, vol.~70, no.~1, pp.~39--45, 1999.

\bibitem{coutino2017near}
M.~Coutino, S.~P. Chepuri, and G.~Leus, ``Near-optimal sparse sensing for
  gaussian detection with correlated observations,'' {\em arXiv preprint
  arXiv:1710.09676}, 2017.

\bibitem{contino2017near}
M.~Coutino, S.~Chepuri, and G.~Leus, ``Near-optimal greedy sensor selection for
  mvdr beamforming with modular budget constraint,'' in {\em Signal Processing
  Conference (EUSIPCO), 2017 25th European}, pp.~1981--1985, IEEE, 2017.

\bibitem{sviridenko2004note}
M.~Sviridenko, ``A note on maximizing a submodular set function subject to a
  knapsack constraint,'' {\em Operations Research Letters}, vol.~32, no.~1,
  pp.~41--43, 2004.

\bibitem{schrijver2003combinatorial}
A.~Schrijver, {\em Combinatorial optimization: polyhedra and efficiency},
  vol.~24.
\newblock Springer Science \& Business Media, 2003.

\bibitem{nemhauser1978analysis}
G.~L. Nemhauser, L.~A. Wolsey, and M.~L. Fisher, ``An analysis of
  approximations for maximizing submodular set functions—i,'' {\em
  Mathematical Programming}, vol.~14, no.~1, pp.~265--294, 1978.

\end{thebibliography}

\end{document}